\documentclass[10pt]{amsart}

\usepackage{amssymb,amsmath,euscript}
\usepackage{mathtools}
\usepackage{breqn}
\usepackage{breqn}
\usepackage{tikz}
\usepackage{tikz}
\usepackage{tkz-berge}
\usepackage{enumitem}

\usepackage{}

\newtheorem{theorem}{Theorem}[section]
\newtheorem{lemma}[theorem]{Lemma}

\newtheorem{corollary}[theorem]{Corollary}
\newtheoremstyle{case}{}{}{}{}{}{:}{ }{}
\theoremstyle{case}

\theoremstyle{definition}
\newtheorem{definition}[theorem]{Definition}
\newtheorem{example}[theorem]{Example}

\theoremstyle{remark}
\newtheorem{remark}[theorem]{Remark}
\def\Fq{{\mathbb F}_q}
\def\a{{\alpha}}

\def\PP{{\mathbb P}}

\newcommand{\Ev}{\operatorname{Ev}}

\newcommand{\supp}{\mathrm{Supp}}

\def \Glm {G_{\ell, m}}

\def \A {\mathbf{A}_\bullet}

\def \Ilm {\mathbb{I}(\ell, m)}
\def \Olm {\Omega_{\a}(\ell, m)}
\def \OA {\Omega(\A)}
\def \C {C(\ell, m)}

\def \P {\overline{P}}

\def \Picl {\overline{P}^{(i)}}
\def \Piicl {\overline{P}^{(i-1)}}

\begin{document}
	
	\title[]{Point-line incidence on Grassmannians and majority logic decoding of Grassmann codes}
	\author{Peter Beelen }
\address{Department of Applied Mathematics and Computer Science,\newline \indent
	Technical University of Denmark,\newline \indent
	Matematiktorvet 303B, 2800 Kgs. Lyngby, Denmark.}
\email{pabe@dtu.dk }
	\author{Prasant Singh}
\address{Department of Mathematics and Statistics,\newline \indent
	University of Tromsø,\newline \indent
	Hansine Hansens veg 18, 9019, Norway.}
\email{psinghprasant@gmail.com}

	\date{\today}
	
\begin{abstract}
In this article, we consider the decoding problem of Grassmann codes using majority logic. We show that for two points of the Grassmannian, there exists a canonical path between these points once a complete flag is fixed. These paths are used to construct a large set of parity checks orthogonal on a coordinate of the code, resulting in a majority decoding algorithm.
\end{abstract}

%

\maketitle
\section{Introduction}

Let $q$ be a prime power and let $\Fq$ be a finite field with $q$ elements.
The Grassmannian $\Glm$ is an algebraic variety whose points correspond to $\ell$-dimensional subspaces of a fixed $m$ dimensional space over $\Fq$.
Corresponding to a projective variety, one can associate a linear code in a natural way using the points of the variety as a projective system \cite{TVN}.
The codes $\C$ associated in this way to the Grassmannians $\Glm$ are known as Grassmann codes.
Grassmann codes were first studied by Ryan \cite{Ryan, Ryan2} over the binary field and later by Nogin \cite{Nogin} over any finite field.
There is was shown that $\C$ is and $[n, k, d]$ code where
\begin{equation}
\label{eq: parameters}
n= {m \brack \ell}_q,\; k=\binom{m}{\ell},\text{ and } d=q^{\ell(m-\ell)},
\end{equation}
where ${m \brack \ell}_q$ is the Guassian binomial coefficient given by
\[
{m\brack \ell}_q:=\dfrac{(q^m-1)(q^{m-1}-1)\cdots(q^{m-\ell+1}-1)}{(q^\ell-1)(q^{\ell-1}-1)\cdots(q-1)}.
\]
These codes have been an object of study ever since they were discovered.
For example, Nogin \cite{Nogin, Nogin1} determined the weight distribution of the Grassmann codes $C(2, m)$ and $C(3, 6)$. Kaipa, et al.~\cite{KP} determined the weight distribution of the Grassmann code $C(3,7)$. Several initial and final generalized Hamming weights of $\C$ are known as well \cite{GL, GPP, Nogin}. Also variants of Grassmann codes, called affine Grassmann codes, obtained by only taking the points in an affine part of the Grassmann variety in the projective system, were studied \cite{BGH}.

However, as far as the efficient decoding of Grassmann codes is concerned, not much is known apart from an approach using permutation decoding \cite{GP,KV} leading to an algorithm capable of correcting up to $d/\binom{m}{\ell}-1$ errors.
In this article we give a decoding algorithm for Grassmann codes $\C$ based on (one-step) majority logic decoding.
A key ingredient is that the dual Grassmann code $\C^\perp$ is a linear code of minimum distance three.
Using ingredients from \cite{BGH2}, it was shown in \cite{BP}, that the weight three parity checks generate $C^\perp$.
This gives the Grassmann code $\C$ an LDPC-like structure and  majority logic decoding is a method used for example in \cite[Ch. 17]{LC} to correct errors for such codes.
Moreover, majority logic decoding has been used to give a decoding algorithm for binary Reed--Muller codes \cite[Th. 20, Ch. 13.7]{MS}, which can be seen as special cases of affine Grassmann codes.
In this article we study to which extent one-step majority logic decoding can be used for Grassmann codes.
In order to do this, we construct sets of parity checks orthogonal on every coordinate of the code.
An essential ingredient in this construction, is the study of paths between points on the Grassmannian, which forms an important part of this paper.
Finally we show that the resulting decoder can correct approximately up to $d/2^{\ell+1}$ errors for a fixed $\ell$ and $q$ tending to infinity. For a fixed $\ell$ and $q$ and $m$ tending to infinity, we can correct up to $M_q(\ell)d/2^{\ell+1}$, where
\begin{equation}\label{eq:Mqell}
M_q(\ell):=
\begin{cases}
\prod_{i=1}^\ell \frac{q^i}{q^i-1}& \text{if } q \text{ is even},\\
\prod_{i=1}^{\ell-1} \frac{q^i(q-1)}{q^{i+1}-1}& \text{if } q \text{ is odd}.
\end{cases}
\end{equation}
This performance compares favourably to previously known efficient decoders for $\C$.

\section{Preliminaries}
 We begin this section with recalling the definitions of Grassmann and Schubert varieties. We give the construction of the Grassmann codes, linear codes associated to Grassmann varieties and recall the parameters of these codes. We define what we call a line in Grassmannians and state a result that classify all these lines in terms of linear subspace of the vector space. For the sake of completeness, we recall some notions and results related to lines in Grassmannian and Grassmann codes are given. These are the results that we will be using in next two sections of this article.

 As in introduction, let $\Fq$ be a finite field with $q$ elements where $q$ is a prime power and let $V=\Fq^m$ be the vector space over $\Fq$ of dimension $m$. Let $\ell\leq m$ be a positive integer. The Grassmannian $\Glm$ of all $\ell$-planes of $V$ is defined by
 \[
 \Glm:=\{P\subseteq V| \;P \text{ is a subspace of V and }\dim P=\ell\}.
 \]
Note that, when $\ell=1$, the Grassmannian $G_{1, m}$ is the projective space $\mathbb{P}(\Fq)^{m-1}$. In general, the Grassmannian $\Glm$ can be embedded into a projective space $\PP^{{m\choose \ell}-1}$ via the Pl\"ucker map. More precisely, let $\Ilm$ be the set defined by
\begin{equation}
\label{eq: Ilm}
\Ilm=\{\a=(\a_1,\ldots, \a_\ell): \; 1\leq\a_1<\cdots <\a_\ell\leq m\}.
\end{equation}

Fix some order on $\Ilm$ and for every point $P\in\Glm$, let $M_P$ be an $\ell\times m$ matrix whose rows forms a basis of $P$. The Pl\"ucker map is the map

$$
Pl: \Glm\to\PP^{{m\choose \ell}-1}\;\text{ defined by } P\mapsto [p_{\a}(M_P)]_{\a\in\Ilm}
$$
where $\a^{\it th}$coordinate, $p_{\a}(M_P)$, is the  $\ell\times \ell$ minor of the matrix $M_P$ labeled by columns $\a$. It is well known that the Pl\"ucker map $Pl$ this is a well defined, injective map. Moreover, the image of the Grassmannian $\Glm$ is a projective algebraic subset of $\PP^{{m\choose \ell}-1}$. It is not hard to see that the cardinality of the  Grassmannian $\Glm$ is given by the Gaussian binomial coefficient  ${m\brack \ell}_q$. Further, $\Glm\subset\PP^{{m\choose \ell}-1}$ can be defined as the common zeroes of the Pl\"ucker polynomials, which are certain irreducible quadratic polynomials.  Hence, the Pl\"ucker map embeds $\Glm$ non-degenerately into $\PP^{{m\choose \ell}-1}$. In other words, $\Glm$ does not lie on any hyperplane in $\PP^{{m\choose \ell}-1}$. Moreover, using duality one can see that $\Glm$ and $G_{m-\ell,m}$ are isomorphic varieties. Therefore we will assume throughout in this article that $\ell \le m-\ell.$
For a more detailed description over Grassmannians and their properties, we refer to \cite{KL, Manivel}.

Next, we recall the definition of Schubert varieties, which are certain subvarieties of  Grassmannians $\Glm$. Let $\a\in\Ilm$ and $\A =(A_1, \ldots, A_\ell)$ be a {\it partial flag of dimension sequence $\a$} or in other words $A_1\subset A_2\subset\cdots\subset A_\ell$ is a sequence of subspaces of $V$ with $\dim A_i=\a_i$ for every $1\leq i\leq \ell$. The {\it Schubert variety} corresponding to this partial flag is defined by
\[
\Omega(\A)=\{P\in\Glm|\; \dim(P\cap A_i)\geq i \text{ for every }i\}.
\]
Schubert varieties are algebraic subvarieties of the Grassmannian $\Glm$. They can seen as the intersection of the Grassmannian and certain coordinate hyperplanes (see \cite{KL, Manivel} for details). A priori, it seems that the variety $\OA$ depends on the partial flag $\A.$ but its geometry depends only on the dimension sequence $\a$. To be precise, if $\mathbf{B}_\bullet$ is another partial flag of dimension sequence $\a$ then there exists an automorphism of $\PP^{{m\choose \ell}-1}$ mapping $\Glm$ onto itself and mapping $\OA$ onto $\Omega(\mathbf{B}_\bullet)$. Therefore we will use the notation $\Olm$ to denote a Schubert variety $\OA$, where $\A$ is a partial flag of dimension sequence $\a$. The set $\Ilm$ is equipped with a natural partial order, call the Bruhat order, which is defined by: for $\a,\;\beta\in\Ilm$ we say that $\beta\le \a$ if and only if $\beta_i\leq \a_i$ for every $1\le i\le\ell$. It is well known \cite[Th. 1]{GT} that the cardinality of the Schubert variety $\Olm$ is given by
\[
|\Olm|=\sum\limits_{\beta\leq \a}q^{\delta(\beta)}, \text{ where } \delta(\beta)=\sum_{i=1}^{\ell}(\beta_i-i).
\]
Note that, if we take $\a=(m-\ell+1,\ldots, m)$, then we have $\beta\le \a$ for any $\beta\in\Ilm$. Hence the corresponding Schubert variety $\Olm$ is the full Grassmannian $\Glm$ in this case. This gives us
\begin{equation}
\label{eq: cardschub}
|\Glm|={m\brack\ell}_q =\sum\limits_{\beta\in\Ilm}q^{\delta(\beta)}.
\end{equation}

Now we are ready to recall the construction of Grassmann codes. Let us fix some representatives $\{P_1, P_2,\ldots, P_n\}$ of the points of $\Glm$ in some fixed order, where $n={m\brack \ell}_q$. Let ${\bf X}=(X_{ij})$ be an $\ell\times m$ matrix in variables $X_{ij}$. 
 For any $\a\in \mathbb{I}(\ell, m)$,  let $X_{\a}$ be the $\ell\times \ell$ minor of   ${\bf X}$ with column index $\a_1, \a_2,\ldots, \a_\ell$. Let $\Fq[X_{\a}:\a\in\mathbb{I}(\ell, m)]_1$ be the vector space of linear polynomials in $X_{\a}$. Consider the evaluation map
\[
\Ev: \Fq[X_{\a}:\a\in\mathbb{I}(\ell, m)]_1\to \Fq^n \text{ defined by } f\mapsto c_f=(f(P_1), \ldots, f(P_n)).
\]
The image of this evaluation map is known as the Grassmann code and we denote this code by $\C$. Note that the length of this code is given by $n= {m \brack \ell}_q$. Further, since the Pl\"ucker map is non-degenerate, it is easy to see that the dimension of the code $\C$ is $\binom{m}{\ell}$. The minimum distance of this code was first determined by Ryan \cite{Ryan, Ryan2}  over  a binary field, and by Nogin \cite{Nogin} for any $q$. They proved that the Grassmann code $\C$ is an $[n,k,d]_q$ linear code, with parameters as in equation \eqref{eq: parameters}.

From the construction it is clear that the coordinates of a codeword of $\C$ can be indexed by the points of $\Glm$. Therefore, we can interpret the support of a codeword $c\in\C$ as a set consisting of points from $\Glm$. To be precise, if $c=c_f\in\C$ is any codeword then we write the support of $c$ as
\[
\supp(c)=\{P\in\Glm: f(P)\neq 0\}.
\]
In the same way, the support of a codeword from $\C^\perp$ will be viewed as a subset of $\Glm$.

Later we will need that the automorphism group of a Grassmann code $\C$ acts transitively on the set of coordinates. This follows easily, since $\mathrm{GL}(m,\Fq)$ acts transitively on $\ell$-dimensional subspaces of $V$. For a full description of the automorphism group of $\C$, see \cite[Th. 3.7]{GK}.

Now, let us describe {\it lines} in $\Glm$, which we will use extensively in the next sections.
In principle, a line in the Grassmannian $\Glm \subset \PP^{{m\choose \ell}-1}$, is simply a line in the projective space $\PP^{{m\choose \ell}-1}$ that is contained in $\Glm$.
However, such lines have a well-known alternative description, which we will use as a definition  \cite[Ch. 3.1]{MP}.
\begin{definition}
Let $U \subset W$ be two subspaces of $V$ of dimensions $\ell-1$ and $\ell+1$ respectively. Then we define
\[
L(U, W):=\{P\in\Glm: U\subset P\subset W\}.
\]
\end{definition}
Note that in this definition, one should identify $\ell$-dimensional subspaces of $V$ with their images in $\PP^{{m\choose \ell}-1}$ under the Pl\"ucker map.
The next lemma is a simple consequence of the definition of a line on the Grassmannian.
\begin{lemma}\cite[Lemma 3]{GPP}
	\label{lemma: line PQ} Let $P$ and $Q$ be two distinct points of the Grassmannian $\Glm$. The following are equivalent:
	\begin{enumerate}
		\item $P$ and $Q$ lie on a line in $\Glm$,
		\item $\dim(P\cap Q)=\ell-1$,
		\item $\dim(P+Q)=\ell+1$.
	\end{enumerate}
\end{lemma}

Dually, it is also not hard to determine whether or not two distinct lines intersect.

\begin{lemma}
	\label{lemma: linesintersect} Let $L(U_1,W_1)$ and $L(U_2,W_2)$ be two distinct lines on the Grassmannian $\Glm$. Then these two lines intersect if and only if one of the following is satisfied:
	\begin{enumerate}
		\item $U_1=U_2$ and $\dim (W_1 \cap W_2) = \ell$,
		\item $W_1=W_2$ and $\dim (U_1 + U_2) = \ell$,
		\item $U_1 \neq U_2$, $W_1 \neq W_2$, and $U_1+U_2=W_1 \cap W_2.$
	\end{enumerate}
In first two cases, the intersection point is $W_1 \cap W_2$, $U_1 + U_2$ respectively. In the third case the intersection point is $U_1+U_2$ (which equals $W_1 \cap W_2$).
\end{lemma}
\begin{proof}
	  It is not hard to see that if (1), (2), or (3) is satisfied, then the lines $L(U_1,W_1)$ and $L(U_2,W_2)$ intersect in the indicated point. Conversely, suppose that $L(U_1,W_1)$ and $L(U_2,W_2)$ intersect. In this case there exist an $\ell$-dimensional space $P$ satisfying $U_1 \subset P \subset W_1$ and $U_2 \subset P \subset W_2.$ If $U_1 \neq U_2$ and $W_1 \neq W_2$, then $U_1+U_2 \subseteq P  \subseteq W_1 \cap W_2$, but equality needs to hold as $\dim(U_1+U_2) \ge \ell \ge \dim(W_1 \cap W_2).$
\end{proof}

\begin{example}
	\label{ex: schub}
	Lines in the Grassmannian $\Glm$ can also be described as a Schubert variety. More precisely, let $L=L(U, W)$ be a line in the Grassmannian. Let $\A=(A_1,\ldots, A_\ell)$ be a partial flag with the dimension sequence  $\a=(1, \ldots, \ell-1,\ell+1)$ where $A_{\ell}=W$, $A_{\ell-1}= U$ and for any $i$ between $1 $ to $\ell-2$ $A_i$ are $i$ dimensional subspaces of $U$ satisfying $A_i\subset A_{i+1}$. Then
	\[
	L(U, W) = \OA
	\]
\end{example}
The following notion of injection distance between two points $P,Q \in \Glm$ is defined in \cite[Def. 2]{SK}.

\begin{definition}
Let $P,Q \in \Glm$ be given. The injection distance between $P$ and $Q$ is defined by $\mathrm{dist}(P,Q):=\ell-\dim(P \cap Q).$	
\end{definition}

In particular Lemma \ref{lemma: line PQ} implies that two distinct points of the Grassmannian lie on a line if and only if they are at distance one. In the next lemma we quote a result from \cite{BCN} in which the number of points at distance $i$ from a given point $P$ was determined.

\begin{lemma}\cite[Lemma 9.3.2]{BCN}
	\label{lemma: picl}
	Let $P \in \Glm$ be given. For any $0\leq i\leq \ell$ the cardinality of the set $\{Q \in \Glm \mid \mathrm{dist}(P,Q)=i\}$ is given by
	\[
	q^{i^2} {\ell\brack i}_q{m-\ell\brack i}_q.
	\]
\end{lemma}

For future reference, we state and prove the following lemma, where an alternative expression for the cardinality of $\{Q \in \Glm \mid \mathrm{dist}(P,Q)=i\}$ is given:

\begin{lemma}\label{lem:formula}
For any $1 \le i \le \ell$ the following identity holds:
	\[
\mathop{\sum_{\ell\ge r_1>\cdots >r_i\ge 1}}_{1\le s_1<\cdots <s_i\le m-\ell}\prod_{j=1}^{i}q^{\ell+i-r_j+s_j-1} = q^{i^2} {\ell\brack i}_q{m-\ell\brack i}_q.
\]
\end{lemma}
\begin{proof}
Let $\mathcal R(i,\ell)$ be the set of all $i$-tuples $\mathbf{r}=(r_1,\ldots, r_i) \in \mathbb{Z}^i$ satisfying $\ell\ge r_1>\cdots >r_i\ge 1$. Similarly, let $\mathbb I (i, m-\ell)$ be the set defined in equation \eqref{eq: Ilm}.
 Further, write $a_j=\ell-r_j+1$ and $\mathbf{a}=(a_1,\ldots, a_i)$. Note $\mathbf{r} \in \mathcal R(i,\ell)$ if and only if $\mathbf{a} \in \mathbb I (i, \ell)$ Then we have
\begin{eqnarray*}
	\mathop{\sum_{\mathbf{r}\in\mathcal R(i,\ell)}}_{\mathbf{s} \in \mathbb I (i, m-\ell) }\prod_{j=1}^{i}q^{\ell-r_j+s_j-1}&=& \mathop{\sum_{\mathbf{a}\in \mathbb I (i, \ell)}}_{\mathbf{s} \in \mathbb I (i, m-\ell)}\prod_{j=1}^{i}q^{a_j+s_j-2}
	\\
	&=& \left(\sum_{\mathbf{a} \in \mathbb I (i, \ell)}q^{\sum_{j=1}^{i}(a_j-1)}\right)\left(\sum_{\mathbf{s} \in \mathbb I (i, m-\ell)}q^{\sum_{j=1}^{i}(s_j-1)}\right)\\
	&=& \left(\sum_{\mathbf{a} \in \mathbb I (i, \ell)}q^{{i\choose 2}}\cdot q^{\delta(\mathbf{a})}\right)\left(\sum_{\mathbf{s} \in \mathbb I (i, m-\ell)}q^{{i\choose 2}}\cdot q^{\delta(\mathbf{s})}\right)\\
	&=& q^{i^2-i}{\ell\brack i}_q{m-\ell\brack i}_q.
\end{eqnarray*}
Here we used equation \eqref{eq: cardschub} in the final equality. The lemma now follows.
\end{proof}

Given $P,Q \in \Glm,$ we say that a sequence of distinct lines $L_1,\dots,L_i$ connects $P$ to $Q$ if $P \in L_1$, $Q \in L_i$ and if for all $1 \le j < i$, the intersection $L_j \cap L_{j+1}$ is not empty.
Then two points $P$ and $Q$ of the Grassmannian are at distance $i$ if and only if there exists a sequence of  $i$ lines $L_1,\dots,L_i$ on the Grassmannian connecting $P$ to $Q$ and no sequence consisting of fewer than $i$ lines connecting $P$ to $Q$ exists. This reformulation of the distance between $P$ and $Q$ used in \cite[Ch.9.3]{BCN} when discussing Grassmann graphs.
%
%
%
We conclude this section by stating the following result from \cite[Thm. 24]{BP} that indicates the key role of lines on Grassmannians in understanding parity checks and hence decoding of $\C$.

\begin{theorem}
	\label{thm: BP}
	The minimum distance of the dual Grassmann code $\C^\perp$ is three. Further, the three points of $\Glm$ corresponding to the support of a minimum weight codeword of $\C^\perp$, lie on a line in the Grassmannian. Conversely, any three points on a line in $\Glm$, form the support of some minimum weight codeword in $\C^\perp$.
\end{theorem}

\section{Geometry of lines on Grassmannians}

In this section we will study the geometry of the lines introduced in the previous section more closely. The notion of distance motivates the following:

\begin{definition}
\label{def: ith closure}
Let $P\in\Glm$ be a point and $i$ be an integer satisfying $0\leq i\leq \ell$. The $i^{\it th}$ closure $\Picl$ of $P$ in $\Glm$ is defined by
\[
\Picl :=\{Q\in\Glm: \mathrm{dist}(P, Q)\leq i\}.
\]
\end{definition}
One can think of $\Picl$ as a ball of radius $i$ and center $P$ within $\Glm.$
Alternatively, one can define
\begin{eqnarray*}
\Picl&=& \{Q\in\Glm: \dim (P \cap Q)\geq \ell-i\}\\
&=& \{Q\in\Glm: \dim (P + Q)\leq \ell+i\}:
\end{eqnarray*}

We extend the definition of $\Picl$  by setting $\Picl =\emptyset$ for any negative integer $i$ and $\Picl= \Glm$ for $i\geq \ell+1$.  Note that $\overline{P}^{(0)}= \{P\}$  
 and $\overline{P}^{(\ell)}=\Glm$. Geometrically, $\Picl$ is the collection of all points $Q$ of the Grassmannian connected to $P$ by a sequence of at most $i$ lines on the Grassmannian. The next lemma shows that the sets $\Picl$ are certain Schubert subvarietes of  $\Glm$. 
 \begin{lemma}
 	\label{lemma: closure as schubert}
 	Let $P\in\Glm$ be an arbitrary point and $i$ be an integer satisfying $0\leq i\leq \ell$. The $i^{\it th}$ closure $\Picl$ of $P$ in the Grassmannian $\Glm$ is the Schubert variety $\Omega_\a(\ell, m)$, where $\a=(i+1,i+2,\ldots,\ell,m-i+1,m-i+2,\ldots, m)$
 \end{lemma}
 \begin{proof}
Fix a partial flag $A_1\subset A_2\subset\cdots\subset A_\ell$ in $V$ where for every $j$ satisfying $1\leq j\leq \ell-i$, $A_j$ is a subspace of $P$ of dimension $i +j$ and for $\ell-i+1\leq j\leq \ell$, $A_j$ is any subspace of $V$ containing $P$ and of dimension $m-\ell+j$. Note that $A_{\ell-i}= P$ and that $Q \in \Omega(\A)$ if and only if $\dim(P \cap Q) \ge \ell-i$. Hence
 \[
 \Omega(\A)=\Picl.
 \]
 \end{proof}

  Note that for every $0\leq i\leq \ell$ we have $\overline{P}^{(i-1)} \subset \Picl$ and that the Grassmannian $\Glm$ is the disjoint union of sets $ \Picl \setminus \overline{P}^{(i-1)}$. More precisely,
 \begin{equation}\label{eq:pic0}
 \Glm= \bigsqcup\limits_{i=0}^{\ell}\left( \Picl \setminus \overline{P}^{(i-1)}\right).
 \end{equation}
Using Lemma \ref{lemma: picl}, one immediately obtains the following:
\begin{equation}\label{eq:picl}
	|\Picl \setminus \overline{P}^{(i-1)}| = q^{i^2} {\ell\brack i}_q{m-\ell\brack i}_q.
\end{equation}

 Next we are going to discuss paths between two points in $\Glm$.

\begin{definition}
	\label{def: path}
	Let $Q\in\Picl\setminus\Piicl$ be a point. A path from $P$ to $Q$ is sequence $\mathcal{P}=(Q_0=P, Q_1,\ldots, Q_{i-1}, Q_i=Q)$ of $i+1$ points in $\Glm$ satisfying
	\[
	\mathrm{dist}(P, Q_t)=t,\;\mathrm{dist}(Q_t, Q_{t+1})=1 \text{ and }	\mathrm{dist}(Q_t, Q)=i-t,\; \forall\; 1\leq t\leq i-1.
	\]
\end{definition}
Note that this definition is equivalent of saying that there are $i$-lines $L(U_t, W_t)$ for $1\le t\le i$ connecting $P$ to $Q$. In this case $Q_t$ is the intersecting point of lines $L(U_t, W_t)$ and $L(U_{t+1}, W_{t+1})$ for every $1\le t\le i-1$.
\begin{lemma}
	\label{lem: path}
	Let $Q\in\Picl\setminus\Piicl$ be a point and let $\mathcal{P}=(P, Q_1,\ldots, Q_{i-1}, Q)$ be a path from $P$ to $Q$. Then
	\[
	P\cap Q_{t+1}\subset P\cap Q_t\text{ and }P+Q_t\subset P+ Q_{t+1}\quad \forall\; 1\leq t\le i-1.
	\]
	In particular, $P\cap Q\subset Q_t\subset P+Q$ for every $1\le t\le i-1$.
\end{lemma}

\begin{proof}
Let $1\le t\le i-1$ be arbitrary. We claim that $P\cap Q_{t+1}\subset P\cap Q_t$. If this is not true, then as $\dim (P\cap Q_{t+1})=\ell-t-1$, we get $\dim (P\cap Q_t\cap Q_{t+1})\le \ell-t-2$. Hence,
\begin{eqnarray*}
\dim ((P\cap Q_t) +Q_{t+1}) &=& \dim (P\cap Q_t) +\dim (Q_t)-\dim (P\cap Q_t\cap Q_{t+1})\\
&\geq& (\ell-t)+ \ell -(\ell-t-2)\\
&=& \ell+2.
\end{eqnarray*}
On the other hand, 	$(P\cap Q_t) +Q_{t+1}\subseteq Q_t +Q_{t+1}$ and $\dim (Q_t+Q_{t+1})=\ell+1$. This is a contradiction and hence we get  $P\cap Q_{t+1}\subset P\cap Q_t$.

Similarly, if $P+Q_t\subset P+ Q_{t+1}$ is not true then, as $\dim(P+Q_{t+1})=\ell+t+1$, we get $\dim(P+Q_{t+1}+ Q_t)\geq\ell+t+1 +1=\ell+t+2$. On the other hand, we have
$(P+Q_t)\cap Q_{t+1}\supseteq (P\cap Q_{t+1}) +(Q_t\cap Q_{t+1})$. Now as $\dim(Q_t\cap Q_{t+1})=\ell-1$, we get $\dim ((P +Q_t)\cap Q_{t+1})\geq \ell-1$. Since $Q_t$ is a point from the path, by definition we have $\dim (P+Q_t)=\ell+t$. This gives
\begin{eqnarray*}
	\dim ((P+Q_t) +Q_{t+1}) &=& \dim (P + Q_t) +\dim Q_{t+1}-\dim ((P+ Q_t)\cap Q_{t+1})\\
	&\leq& (\ell+t)+ \ell -(\ell-1)\\
	&=& \ell+t +1,
\end{eqnarray*}
which is a contradiction.
\end{proof}	

For the rest of the article we fix a point $P \in \Glm$, an integer $1\le i\le \ell$  and  a complete flag passing through $P$:
\[
(0)=\mathcal{U}_0\subset \mathcal{U}_1 \subset \mathcal{U}_2\subset \cdots\mathcal{U}_{\ell-1}\subset\mathcal{U}_{\ell}=P=\mathcal{W}_{\ell}\subset \mathcal{W}_{\ell+1}\subset \cdots\mathcal{W}_{m-1}\subset\mathcal{W}_{m}=V.
\]
We will now investigate paths satisfying certain condition with respect to this flag.

\begin{definition}
	Let $Q\in\Picl\setminus\Piicl$ be a point. Given a path $\mathcal{P}$  from $P$ to $Q$, say $\mathcal{P}=(P, Q_1,\ldots, Q_{i-1}, Q)$, we define two $i$-tuples $\mathbf{r}(\mathcal{P})=(r_1(\mathcal{P}),\ldots, r_i(\mathcal{P}))$ and $\mathbf{s}(\mathcal{P})=(s_1(\mathcal{P}),\ldots, s_i(\mathcal{P}))$, where for $1\le t\le i$:
	\[
	r_t(\mathcal{P})=\max\{j:\mathcal{U}_{j-1}\subseteq Q_t\}
	\]
	and
	\[
	s_t(\mathcal{P})=\min\{j:Q_t\subseteq \mathcal{W}_{\ell+j}\}.
	\]
\end{definition}
To ease the notation, we will sometimes write $r_t$ and $s_t$ instead of $r_t(\mathcal{P})$ and $s_t(\mathcal{P})$ if the path $\mathcal P$ is fixed.
In the next lemma we will show that these $i$-tuples for a given path from $P$ to a point $Q$ are increasing. More precisely,
\begin{lemma}
	\label{lem: monotone}
	Let $Q\in\Picl\setminus\Piicl$ and a path $\mathcal{P}=(P, Q_1,\ldots, Q_{i-1}, Q)$ from $P$ to $Q$ in $\Glm$ be given. Then the corresponding $i$-tuples $\mathbf{r}(\mathcal{P})$ and $\mathbf{s}(\mathcal{P})$ satisfy
\[
\ell\ge r_1\ge r_2\ge\cdots\ge r_i\ge 1 \text{ and }1\le s_1\le\cdots\le s_i\le m-\ell.
\]

\end{lemma}
\begin{proof}
We only prove the first part involving $\mathbf{r}(\mathcal{P})$. The second part can be shown similarly. Clearly $\ell\ge r_1$. Now let $2\le t\le i$ and let $r_t=j$. By definition, this means $\mathcal{U}_{j-1}\subseteq Q_t$ but $\mathcal{U}_{j}\nsubseteq Q_t$. As $\mathcal{U}_{j-1}\subset P$, we get $\mathcal{U}_{j-1}\subseteq Q_t\cap P$. From Lemma \ref{lem: path} we have $P\cap Q_t\subseteq P\cap Q_{t-1}$. Consequently, $\mathcal{U}_{j-1}\subseteq Q_{t-1}$ and hence $r_{t-1}\geq j$. This completes the proof for the sequence $\mathbf{r}(\mathcal{P})$.
\end{proof}
For any point $Q\in\Picl\setminus\Piicl$ we define some new constants that are going to be very useful in understanding the paths between $P$ and $Q$.
\begin{definition}
	Let $Q\in\Picl\setminus\Piicl$ be a given point. For every $1\le t\le i$ we define
	\[
	\gamma_{t}(Q)=\max\{j: \dim(Q+\mathcal{U}_j)=\ell +i-t\}
	\]
	and
		\[
	\delta_{t}(Q)=\min\{j: \dim(Q\cap\mathcal{W}_{\ell+j})=\ell -i +t\}
	\]
\end{definition}
If from the context the point $Q$ is clear, we will simply write $\gamma_t$ and $\delta_t$.
The constants $\gamma_t$ indicate the jump positions (in reverse order) in the dimension in the sequence of nested subspaces $Q+\mathcal U_0 \subseteq Q+\mathcal U_1 \subseteq \cdots \subseteq Q+\mathcal U_\ell=Q+P.$ Hence $0\le \gamma_i<\gamma_{i-1}<\cdots<\gamma_1$. Moreover $\gamma_1\le \ell-1$, since $\dim(P+Q)=\ell+i$.
Similarly, the constants $\delta_t$ indicate the jump positions in dimension in the sequence of nested subspaces $Q \cap P=Q \cap \mathcal W_\ell \subseteq Q \cap \mathcal W_{\ell+1} \subseteq \cdots \subseteq Q \cap \mathcal W_{m} =Q.$  Hence $1\le \delta_1<\delta_2<\cdots<\delta_i\le m-\ell$. In the next theorem, we will show that for every $Q\in\Picl\setminus\Piicl$ there exist a path such that the corresponding $i$-tuples are strictly increasing. The constants $\gamma_t$ and $\delta_t$ will appear in a natural way. First we need a lemma.

\begin{lemma}\label{lem:pathstrict}
Let $Q\in\Picl\setminus\Piicl$, and recursively define
$$Q_t:=\begin{cases} P & \text{if } t=0,\\ ((Q_{t-1}\cap Q)+\mathcal{U}_{\gamma_t}) + (\mathcal{W}_{\ell+\delta_t}\cap Q) & \text{if } 1 \le t \le i.\end{cases}.$$
Then $\mathcal P=(Q_0,\dots,Q_i)$ is a path from $P$ to $Q$.
\end{lemma}
\begin{proof}
Directly from the definition, we see that $Q_0=P$.
Moroever, note that $\dim(Q+\mathcal{U}_{\gamma_i})=\ell$ and $\dim(\mathcal{W}_{\ell+\delta_i} \cap Q)=\ell$. Hence $Q+\mathcal{U}_{\gamma_i}=Q=\mathcal{W}_{\ell+\delta_i} \cap Q,$ which implies that $Q_i=Q.$

We will now prove with induction on $t$ the claim that for all $0 \le t \le i-1$:
$$ \dim(Q_t)=\ell, \ \dim (P \cap Q_{t})=\ell-t, \ \dim (Q_{t}\cap Q_{t+1})=\ell-1, \text{ and } \dim (Q_{t}\cap Q)=\ell-i+t.$$

If $t=0$, the only nontrivial statement is that $\dim(P \cap Q_1)=\ell-1.$ We have $Q_1=((P\cap Q)+\mathcal{U}_{\gamma_1}) + (\mathcal{W}_{\ell+\delta_1}\cap Q)$. Since $(P\cap Q)+\mathcal{U}_{\gamma_1} \subset P,$ we have
$$P\cap Q_1=((P\cap Q)+\mathcal{U}_{\gamma_1}) + (P \cap \mathcal{W}_{\ell+\delta_1}\cap Q)=((P\cap Q)+\mathcal{U}_{\gamma_1}) + (P \cap Q)=(P\cap Q)+\mathcal{U}_{\gamma_1}.$$
Moreover,
$\dim((P\cap Q)+\mathcal{U}_{\gamma_1})=\dim(P \cap Q)+\dim(\mathcal{U}_{\gamma_1})-\dim(P \cap Q \cap \mathcal{U}_{\gamma_1}).$
Since $P \cap Q \cap \mathcal{U}_{\gamma_1}=Q \cap \mathcal{U}_{\gamma_1}$ and by definition $\dim(Q + \mathcal{U}_{\gamma_1})=\ell+i-1,$ we may conclude that $\dim(P \cap Q_1)=\ell-1.$ Here we computed the dimension $Q \cap \mathcal{U}_{\gamma_1}$ using that $\dim(Q + \mathcal{U}_{\gamma_1})=\ell+i-1$ by the definition of $\gamma_1.$

Now assume that the claim holds for $t-1$. Since $\gamma_t<\gamma_{t-1}$, we get $Q \cap \mathcal{U}_{\gamma_t}\subseteq \mathcal{U}_{\gamma_t} \subset \mathcal{U}_{\gamma_{t-1}}$. The definition of $Q_{t-1},$ implies $\mathcal{U}_{\gamma_{t-1}}\subset Q_{t-1}$. We conclude $Q\cap\mathcal{U}_{\gamma_t}\subset\mathcal{U}_{\gamma_{t-1}}\subset Q_{t-1}.$
Hence inductively we get
\begin{eqnarray}\label{eq:distance1}
\dim((Q_{t-1}\cap Q)+\mathcal{U}_{\gamma_t})&=& \dim (Q_{t-1}\cap Q) +\dim  \mathcal{U}_{\gamma_t}-\dim ((Q_{t-1}\cap Q)\cap\mathcal{U}_{\gamma_t}) \notag\\
&=& (\ell-i +t-1) + \gamma_t - \dim (Q\cap \mathcal{U}_{\gamma_t} ) \notag\\
&=& (\ell-i +t-1) + \gamma_t - (\gamma_t-i+t) \\
&=& \ell-1. \notag
\end{eqnarray}
By definition of $Q_{t-1}$ we have $\mathcal{W}_{\ell+\delta_{t-1}}\cap Q\subset Q_{t-1}\cap Q$ and using the induction hypothesis, both are of dimension $\ell-i +t-1$. Therefore $\mathcal{W}_{\ell+\delta_{t-1}}\cap Q= Q_{t-1}\cap Q$.
As $\delta_t> \delta_{t-1}$, we get $\mathcal{W}_{\ell-\delta_{t-1}}\subset \mathcal{W}_{\ell-\delta_{t}}$ and hence
\begin{equation}\label{eq:subsetW}
(Q_{t-1}\cap Q)+\mathcal{U}_{\gamma_t} \subset (\mathcal{W}_{\ell+\delta_{t-1}}\cap Q)+\mathcal{U}_{\gamma_t} \subset \mathcal{W}_{\ell+\delta_{t}}.
\end{equation}
Consequently
\[
((Q_{t-1}\cap Q)+\mathcal{U}_{\gamma_t})\cap (\mathcal{W}_{\ell+\delta_t}\cap Q)=((Q_{t-1}\cap Q)+\mathcal{U}_{\gamma_t})\cap Q.
\]
On the other hand $((Q_{t-1}\cap Q)+\mathcal{U}_{\gamma_t})\cap Q= (Q_{t-1}\cap Q) +(\mathcal{U}_{\gamma_t}\cap Q)$. But the right-hand side is equal to $Q_{t-1}\cap Q$ as $Q_{t-1}\supseteq\mathcal{U}_{\gamma_{t-1}}\supseteq \mathcal{U}_{\gamma_{t}}$. Putting all this together, we get
\begin{eqnarray*}
	\dim Q_t&=& (\ell-1) +(\ell-i+t)-\dim ((Q_{t-1}\cap Q)+\mathcal{U}_{\gamma_t})\cap(\mathcal{W}_{\ell+\delta_t}\cap Q)\\
	&=& (\ell-1) +(\ell-i+t) - \dim (Q_{t-1}\cap Q )\\
	&=& \ell.
\end{eqnarray*}
This proves the first part of the claim that $\dim(Q_t)=\ell.$

The definition of $Q_t$ implies that $((Q_{t-1}\cap Q)+\mathcal{U}_{\gamma_{t}})\cap P\subseteq Q_t\cap P$. Now, using the definition of $Q_{t-1}$, we obtain $P\cap Q_{t-1} \supset P \cap \mathcal{W}_{\ell-\delta_{t-1}} \cap Q = P\cap Q.$ Hence, we may conclude that
$((Q_{t-1}\cap Q)+\mathcal{U}_{\gamma_{t}})\cap P=(P\cap Q)+ \mathcal{U}_{\gamma_{t}}.$ Moreover, $\dim ((P \cap Q)+\mathcal{U}_{\gamma_{t}} )= \ell-t$, since $\dim(Q+\mathcal{U}_{\gamma_t})=\ell-i+t$. Combining the above, we get $\dim (P\cap Q_t) \ge \ell-t$ and consequently $\mathrm{dist}(P, Q_t)\le t$. Similarly, as $\mathcal{W}_{\ell+\delta_t}\cap Q\subset Q_t\cap Q$, one obtains $\dim (Q_t\cap Q)\ge \ell-i+t$ and hence $\mathrm{dist}(Q, Q_t)\le i-t$. As $\mathrm{dist}(P, Q) =i$ we conclude $\mathrm{dist}(P, Q_t)= t$ and $\mathrm{dist}(Q, Q_t)= i-t$.
This proves that $\dim (P \cap Q_{t})=\ell-t$ and $\dim (Q \cap Q_{t})=\ell-i+t.$

What remains to be shown is that $\dim(Q_t \cap Q_{t+1})=\ell-1.$ Since $(Q \cap Q_t)+\mathcal{U}_{\gamma_{t+1}} \subset Q_t$, we obtain that
$$Q_t \cap Q_{t+1}=((Q \cap Q_t)+\mathcal{U}_{\gamma_{t+1}})+(\mathcal{W}_{\ell+\delta_{t+1}} \cap Q \cap Q_t)=(Q \cap Q_t)+\mathcal{U}_{\gamma_{t+1}}.$$ Similarly as in equation \eqref{eq:distance1}, we can now show that $\dim(Q_t \cap Q_{t+1})=\ell-1.$ This proves the claim.

The claim immediately implies that $\mathcal P$ is a path from $P$ to $Q$.
\end{proof}

\begin{theorem}
\label{thm: strictineq}
For every $Q\in\Picl\setminus\Piicl$, the $i$-tuples $\mathbf{r}(\mathcal P)$ and $\mathbf{s}(\mathcal P)$ corresponding to the path $\mathcal P$ constructed in Lemma \ref{lem:pathstrict}, are given by
\[
r_t=\gamma_t(Q)+1 \text{ and } s_t=\delta_t(Q), \text{ for } 1 \le t \le i.
\]
In particular these $i$-tuples satisfy:
	\[
	\ell\ge r_1> r_2>\cdots> r_i\ge 1 \text{ and } 1\le s_1<\cdots < s_i\le m-\ell.
	\]
\end{theorem}
\begin{proof}
We will use the path $\mathcal P$ constructed in Lemma \ref{lem:pathstrict} and determine its $i$-tuples $\mathbf{r}(\mathcal P)$ and $\mathbf{s}(\mathcal P).$
First, we claim that $r_t= \gamma_{t}+1$. Recall that
$$
r_t=\max\{j: \mathcal{U}_{j-1}\subseteq Q_{t}\}.
$$
By definition, we have $\mathcal{U}_{\gamma_{t}}\subset Q_t$. This gives $r_t\ge \gamma_t+1$. On the other hand if $\mathcal{U}_{\gamma_{t}+1}\subset Q_t$ then $\mathcal{U}_{\gamma_{t}+1} +Q\subseteq Q_t+Q$. But we also have $\dim (Q_t +Q)=\ell +i-t$ and by definition of $\gamma_t$ we get $\dim(\mathcal{U}_{\gamma_{t}+1} +Q)>\dim (\mathcal{U}_{\gamma_{t}} +Q)=\ell+i-t$. But this is a contradiction. This implies $\mathcal{U}_{\gamma_{t}+1}\nsubseteq Q_t$. In particular, $r_t\le\gamma_t+1$ and hence $r_t=\gamma_t+1$ for every $1\le t\le i$. Also, recall that
$$
s_t=\min\{j:Q_t\subset W_{\ell+j}\}.
$$
Using equation \eqref{eq:subsetW}, we know $Q_t \subset \mathcal{W}_{\ell+\delta_t}$ and hence $s_t\le \delta_t$.  Now, if $Q_t\subseteq \mathcal{W}_{\ell+\delta_t-1}$ then $Q_t\cap Q\subseteq \mathcal{W}_{\ell+\delta_t-1}\cap Q$. Note that this gives $\dim(\mathcal{W}_{\ell+\delta_t-1}\cap Q)\ge \ell-i+t$ but by definition of $\delta_t$ we have $\dim(\mathcal{W}_{\ell+\delta_t-1}\cap Q)< \ell-i+t$. This is a contradiction. Hence we get $s_t=\delta_t$ for every $1\le t\le i$. This completes the proof of the theorem.
\end{proof}	

\begin{remark}\label{rem:rQ}
Note that the path $\mathcal P$ constructed in Lemma \ref{lem:pathstrict} only depends on $P$, $Q$ and the flag. Since $P$ and the flag are fixed throughout, we will therefore for this path use the notations $\mathbf{r}(Q)$ and $\mathbf{s}(Q)$ instead of $\mathbf{r}(\mathcal P)$ and $\mathbf{s}(\mathcal P).$
\end{remark}

In the next theorem we will prove that for a given $Q\in\Picl\setminus\Piicl$ there is a unique path $(P, Q_1,\ldots, Q_{i-1}, Q)$ such that the corresponding $i$-tuples $\mathbf{r}=(r_1,\ldots, r_i)$ and $\mathbf{s}=(s_1,\ldots, s_i)$ satisfy the strict inequality condition. This implies in particular that this path has to be the one constructed in Lemma \ref{lem:pathstrict}.

\begin{theorem}
	\label{thm: uniquepath}
	Let $Q\in\Picl\setminus\Piicl$ and let $\mathcal{P}^\prime=(P, Q_1^\prime,\ldots, Q_{i-1}^\prime, Q)$ be an arbitrary path from $P$ to $Q$. Let  $\mathbf{r}(\mathcal{P}^\prime)=(r_1^\prime,\ldots, r_i^\prime)$ and $\mathbf{s^\prime}(\mathcal{P}^\prime)=(s_1^\prime,\ldots, s_i^\prime)$ are corresponding $i$-tuples and suppose that
	$$r_1^\prime>\cdots> r_i^\prime \text{ and } s_1^\prime<\cdots < s_i^\prime.$$
Then $Q_j^\prime=Q_j$ for every $1\le j\le i$, where the $Q_j$ are defined as in Lemma \ref{lem:pathstrict}.	
\end{theorem}

\begin{proof}
	We claim that $r_t^\prime=\gamma_t+1$ for every $1\le t\le i$.
Since $\mathcal{U}_{r_t^\prime-1}\subset Q_t^\prime$, we get $\mathcal{U}_{r_t^\prime-1}+ Q\subset Q_t^\prime+ Q$ and hence $\dim (\mathcal{U}_{r_t^\prime-1}+ Q)\le \dim(Q_t^\prime+Q)= \ell+i-t$. By definition of $\gamma_t$ we get $r_t^\prime-1\le \gamma_t$. Now, if $r_t^\prime-1 < \gamma_t$, then we get $\dim(\mathcal{U}_{r_t^\prime}+ Q)=\dim(\mathcal{U}_{r_k^\prime-1}+ Q)$ for some $k> t$. As $\mathcal{U}_{r_k^\prime-1}\subset \mathcal{U}_{r_t^\prime-1}$, we obtain that $\mathcal{U}_{r_t^\prime-1}+ Q=\mathcal{U}_{r_k^\prime-1}+ Q$. Intersecting both sides of this equality with $P$, we get $\mathcal{U}_{r_t^\prime-1}+ (Q\cap P)=\mathcal{U}_{r_k^\prime-1}+ (Q\cap P)$. By Lemma \ref{lem: path}, we have $P\cap Q\subset Q_k^\prime$ and moreover $\mathcal{U}_{r_k^\prime-1}\subset Q_k^\prime$ by definition of $r_k^\prime$. Hence $\mathcal{U}_{r_t^\prime-1} \subset \mathcal{U}_{r_t^\prime-1}+(Q\cap P) = \mathcal{U}_{r_k^\prime-1}+(Q\cap P) \subseteq Q_k^\prime$, implying $r_k^\prime\geq r_t^\prime$. But this contradicts the strict inequality $r_k^\prime< r_t^\prime$. Therefore, we get $r_t^\prime-1= \gamma_t$.

Similarly, from the definition of $s_t^\prime$ we have $Q_t^\prime\subseteq\mathcal{W}_{\ell+s_t^\prime}$ and this gives $\dim(\mathcal{W}_{\ell+s_t^\prime}\cap Q)\geq \dim(Q_t^\prime \cap Q) = \ell - i+t$. Consequently, $\delta_t\le s_t^\prime$. Now if $\delta_t < s_t^\prime$, then $\dim(\mathcal{W}_{\ell+s_t^\prime}\cap Q)=\dim(\mathcal{W}_{\ell+s_k^\prime}\cap Q)$ for some $k>t$. Then $\mathcal{W}_{\ell+s_t^\prime}\cap Q=\mathcal{W}_{\ell+s_k^\prime}\cap Q$. Adding $P$ both sides and keeping in mind that $P\subset \mathcal{W}_{\ell+j}$ for every $j$, we get $\mathcal{W}_{\ell+s_t^\prime}\cap (P+Q)= \mathcal{W}_{\ell+s_k^\prime}\cap (P+Q) $. Since $Q_k^\prime\subset\mathcal{W}_{\ell+s_k^\prime} $ by definition of $s_k^\prime$ and  $Q_k^\prime\subset P+Q$ by Lemma \ref{lem: path}, we get $Q_k^\prime\subset \mathcal{W}_{\ell+s_j^\prime}$ and consequently, $s_k^\prime\le s_t^\prime$. But this contradicts the strict inequality $s_k^\prime> s_t^\prime$. Hence $s_t^\prime=\delta_t$.

Now, we will show that $Q_t=Q_t^\prime$ for $1\le t \le i$ by induction on $t$. It is enough to prove that for every $1\le t\le i$, $(Q_{t-1}\cap Q)+\mathcal{U}_{\gamma_t}\subseteq Q_t^\prime$ and $\mathcal{W}_{\ell+\delta_t}\cap Q\subseteq Q_t^\prime$.

If $t=1$, then $(Q_{t-1}\cap Q)+\mathcal{U}_{\gamma_t}=(P\cap Q)+\mathcal{U}_{\gamma_1}$, which is contained in $Q_1^\prime$, since $P\cap Q \subset Q_1^\prime$ by Lemma \ref{lem: path} and $\mathcal{U}_{\gamma_1} \subset Q_1^\prime$ by definition of $r_1^\prime$ and the fact that $r_1^\prime-1=\gamma_1$.
Similarly by definition of $s_1^\prime$ and the fact that $s_1^\prime=\delta_1$, we get $\mathcal{W}_{\ell+\delta_1}\cap Q\supseteq Q_1^\prime \cap Q$. Since both spaces have dimension $\ell-i+1$, they are equal. Hence $\mathcal{W}_{\ell+\delta_1}\cap Q\subseteq Q_1^\prime$.

Now for the induction step assume $t>1$ and $Q_j=Q_j^\prime$ for every $1\le j\le t-1$. We have $Q_{t-1}\cap Q)+\mathcal{U}_{\gamma_t}=(Q_{t-1}^\prime\cap Q)+\mathcal{U}_{\gamma_t}.$ Applying Lemma \ref{lem: path} to the path $(Q_{t-1}^\prime,\dots,Q_{i-1}^\prime,Q)$, we see that $Q_{t-1}^\prime\cap Q \subset Q_t^\prime$. Moreover, since $r_t^\prime-1=\gamma_t$, we get $\mathcal{U}_{\gamma_t}\subset Q_t^\prime$. This shows that $Q_{t-1}\cap Q)+\mathcal{U}_{\gamma_t}\subset Q_t^\prime $
Similarly as in the induction basis, by definition of $s_t^\prime$ and the fact that $s_t^\prime=\delta_t$, we get $\mathcal{W}_{\ell+\delta_t}\cap Q\supseteq Q_t^\prime \cap Q$. Since both spaces have dimension $\ell-i+t$, they are equal. Hence $\mathcal{W}_{\ell+\delta_t}\cap Q\subseteq Q_t^\prime$. This concludes the proof.
\end{proof}

 \section{A majority logic decoder for $\C$}

Our aim in this section is to construct a decoder for the Grassmann codes $\C$ that runs in quadratic complexity in the length of the code. In order to do this, we will construct certain ``orthogonal" parity checks of $\C$ and then use the well-known method of majority logic decoding.
 First, we recall what we mean by orthogonal parity checks and how to use them for majority logic decoding. For a general reference on these topics, \cite[Ch 13.7]{MS} for the binary case and \cite[Ch 1]{M} for the $q$-ary case. As usual, we call a codeword of the dual code $\C^\perp$ is a parity check for $\C$.
  \begin{definition}
  	Let $C$ be an $[n, k]$ code. A set $\mathcal J$ of $J$ parity checks of $C$ is said to be orthogonal on the $i^{\it th}$ coordinate if the $J\times n$ matrix $H$ having these $J$ parity checks as rows satisfies the following: 
  	\begin{enumerate}
  		\item Each entry in the $i^{\it th}$ column of $H$ is $1$.
  		\item The Hamming weight of any other column of $H$ is at most $1$, i.e., if $j \neq i$ and the $j^{\it th}$ column of $H$ contains a non-zero entry in the $r^{\it th}$ row, then this is the only non-zero entry in this column.
  	\end{enumerate}
  \end{definition}

Suppose that $c \in C$ is the sent codeword, but that the receiver receives the word $w=c+e$, for some $e=(e_1,\dots,e_n) \in \mathbb{F}_q^n.$ Given a coordinate $i$ and a set $\mathcal J=\{\omega_1,\dots,\omega_J\}$ of parity checks orthogonal on the $i^{\it th}$ coordinate, for each parity check, we define $S_j(w):=\sum_{a=1}^n w_a \omega_{j,a}.$ Note that $S_j(w)=S_j(e)=e_i+\sum_{a=1;a\neq i}^n e_a \omega_{j,a}$.
Now if a clear majority of the $J$ values $S_j(w)-w_i$, where $1 \le j \le J$, equals $-\alpha$, then we define $\widehat{c}_i:=\alpha$, otherwise we set $\widehat{c}_i:=w_i.$ Doing this for each coordinate $i$, results in the decoded word $\widehat{c}:=(\widehat{c}_1,\dots,\widehat{c}_n)$. This procedure of determining $\widehat{c}$ is called majority logic decoding. It is not a priori clear that $\widehat{c}$ is a codeword or if it is, that it is equal to the sent codeword $c$. However, the following theorem from \cite{M} guarantees that $\widehat{c}=c$ as long as the number of errors, i.e., the Hamming weight of $e$ is at most $\lfloor J/2 \rfloor$.
  \begin{theorem}\cite[Ch 1,Thm 1]{M}
  	\label{thm: Maj decod}
  	Let $C$ be an $[n,k]$ code such that for each $1 \le i \le n$, there exists a set $\mathcal J$ of $J$ orthogonal parity checks on the $i^{\it th}$ coordinate. Then the corresponding majority logic decoder corrects up to $\lfloor J/2 \rfloor$ errors.
  \end{theorem}

  To use this theorem for the decoding of Grassmann codes, we need to construct as many orthogonal parity checks as possible for each coordinate. However, as the automorphism group of $\C$ acts transitively on the coordinates, we only need to produce such parity checks for a single fixed coordinate. Then sets of parity checks orthogonal on other coordinates can be obtained immediately. Therefore, for the rest of the article we fix $P\in \Glm$ and will construct parity checks that are orthogonal on the coordinate corresponding to $P$. The starting point of our construction is Theorem \ref{thm: BP}.
  First, note that if we take a line in $\Glm$ passing through $P$ and any two points $Q$ and $R$ different from $P$ on that line, then Theorem \ref{thm: BP} guarantees the existence of a parity check for $\C$ with support corresponding to $P$, $Q$ and $R$.
  Note that if $q=2$, for a given line through $P$, there is a unique choice for $Q$ and $R$, since in that case a line contains exactly three points.
  In this way, we can obtain for each line one parity check of Hamming weight $3$ whose support contains $P$.
  All the parity checks obtained in this way are orthogonal on $P$ as they all are passing through $P$ and any two distinct lines through $P$ only intersect at $P$.
  In this way we get ${\ell\brack 1}_q{m-\ell\brack 1}_q$ many parity checks orthogonal on $P$.
  Before giving the general construction, we illustrate in the next example how are we are going to use the parity checks corresponding to lines through $P$ to increase the set of parity checks orthogonal on $P$.

  \begin{example}\label{ex:g24}
  	Let $V=\mathbb{F}_2^4$ and $G_{2,4}$ be the Grassmannian of all planes of $V$. Let $C(2, 4)$ be the corresponding binary Grassmann code. Then $C(2, 4)$ is a binary $[n, k, d]$ code where
  $$
  	n= {4\brack 2}_2= 35,\quad k= 6, \quad \text{and} \quad d=16.
 $$
  	Now let $\{e_1,\ldots, e_4\}$ be the standard basis of $V$ and $P=\langle e_1, e_2\rangle$. There are ${2\brack 1}_2{2\brack 1}_2=9$ lines in $G_{2,4}$ passing through $P$. Explicitly these lines are $L(U,W)$, where there are three possible choices for $U$, namely $\langle e_1\rangle$, $\langle e_2\rangle$, or $\langle e_1+e_2\rangle$, and three possibilities for $W$, namely $\langle e_1,e_2,e_3\rangle$, $\langle e_1,e_2,e_4\rangle$, or $\langle e_1,e_2,e_3+e_4\rangle$. For example, we have $L(\langle e_1\rangle,\langle e_1,e_2,e_3\rangle)=\{P,\langle e_1,e_3\rangle,\langle e_1,e_2+e_3\rangle\}.$

  	 Each of these nine lines corresponds to a weight three parity check. These parity checks are orthogonal on $P$.
  	 As mentioned before, the three points on these lines form the support of the corresponding parity check.
  	 To increase the number of parity checks orthogonal on $P$, we combine the nine we have found so far with other weight three parity checks in a structured way.
  	 Consider the line $L(\langle e_1\rangle,\langle e_1,e_2,e_3\rangle)$. There are nine lines through $\langle e_1, e_3\rangle$. Let $L(U, W)$ be a line through  $\langle e_1, e_3\rangle$. One can verify directly that if $U\neq \langle e_1\rangle$ and $W\neq \langle e_1, e_2, e_3\rangle$, then the two points on $L(U, W)$ different from  $\langle e_1, e_3\rangle$, lie in $\overline{P}^{(2)}\backslash \overline{P}^{(1)}$. In this way, we get four lines through $\langle e_1, e_3\rangle$ intersecting $\overline{P}^{(1)}$ only at $\langle e_1, e_3\rangle$. Similarly we will get four such lines passing through the third point $\langle e_1, e_2+e_3\rangle$. The lines are given in the figure below. Now, we enumerate the four lines through $\langle e_1, e_3\rangle$, say $m_1=L(\langle e_3\rangle,\langle e_1,e_3,e_4\rangle)$, $m_2=L(\langle e_3\rangle,\langle e_1,e_3,e_2+e_4\rangle)$, $m_3=L(\langle e_1+e_3\rangle,\langle e_1,e_3,e_4\rangle)$, $m_4=L(\langle e_1+e_3\rangle,\langle e_1,e_3,e_2+e_4\rangle)$, as well as the four lines through $\langle e_1, e_2+e_3\rangle$, say $n_1=L(\langle e_2+e_3\rangle,\langle e_1,e_2+e_3,e_4\rangle)$, $n_2=L(\langle e_2+e_3\rangle,\langle e_1,e_2+e_3,e_2+e_4\rangle)$, $n_3=L(\langle e_1+e_2+e_3\rangle,\langle e_1,e_2+e_3,e_4\rangle)$, $n_4=L(\langle e_1+e_2+e_3\rangle,\langle e_1,e_2+e_3,e_2+e_4\rangle)$. Let $\omega$ be the parity check corresponding to the line $L(\langle e_1\rangle,\langle e_1,e_2,e_3\rangle)$, $\omega_i$ be the parity check corresponding to the $i^{\it th}$ line through $\langle e_1, e_3\rangle$ and $\omega_i^\prime$ be the parity check corresponding to the $i^{\it th}$ line through $\langle e_1,e_2 + e_3\rangle$.
  	 For every $i$ the parity check $\omega + \omega_i + \omega_i^{\prime}$ is of weight five. Further, these four weight five parity checks are again orthogonal on $P$ as their supports consists of $P$ and pairwise disjoint sets of four points from $\overline{P}^{(2)}\setminus\P$.
  	 Therefore the set of $9+4=13$ parity checks obtained in this way is orthogonal on $P$.
  	 Note that we can not increase the set of these parity checks any further.
  	 This is simply because the total support of these $13$ parity checks consists of $1 + 9\times 2 +4\times 4=35$ points.
  	 However, $G_{2,4}$ contains exactly that many points, so there is no room for any further parity checks without violating orthogonality.
     Now using the automorphism group, we can for each coordinate produce a set of $13$ parity checks orthogonal on that coordinate.
     Theorem \ref{thm: Maj decod} implies that we can correct up to six errors for $C(2, 4)$ using this approach.
  		
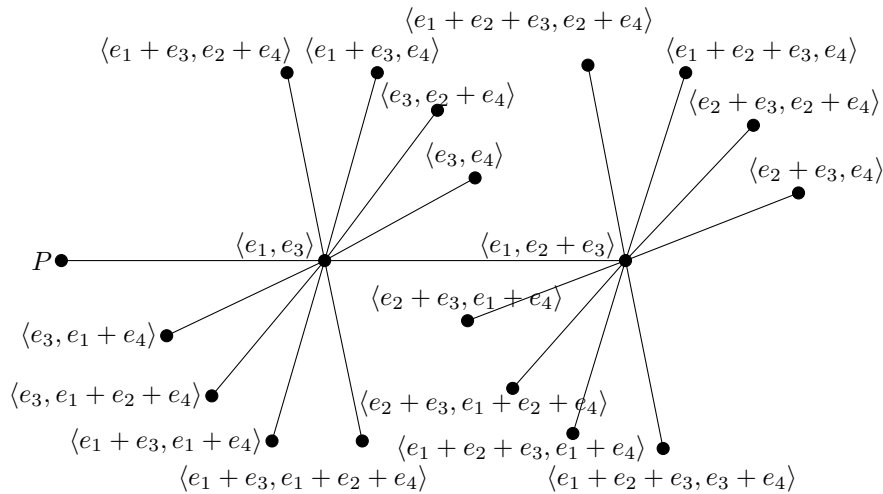
\[
  	\begin{tikzpicture}[baseline=(current bounding box.north), level/.style={sibling distance=50cm}]
  	\tikzstyle{edge} = [draw,thick,-]	
  	\draw[-] (0,0) -- (3.5,0);
  	\draw[-] (3.5,0) -- (7.5,0);
  	\draw[-] (3.5,0) -- (5.5,1.1);
  	\draw[-] (3.5,0) -- (1.4,-1);
  	\draw[-] (3.5,0) -- (5,2);
  	\draw[-] (3.5,0) -- (2,-1.8);
  	\draw[-] (3.5,0) -- (4.2,2.5);
  	\draw[-] (3.5,0) -- (2.8,-2.4);
  	\draw[-] (3.5,0) -- (3,2.5);
  	\draw[-] (3.5,0) -- (4,-2.4);
  	
  	\draw[-] (7.5,0) -- (9.8,.9);
  	\draw[-] (7.5,0) -- (5.4,-.8);
  	\draw[-] (7.5,0) -- (9.2,1.8);
  	\draw[-] (7.5,0) -- (6,-1.7);
  	\draw[-] (7.5,0) -- (8.3,2.5);
  	\draw[-] (7.5,0) -- (6.8,-2.3);
  	\draw[-] (7.5,0) -- (7,2.6);
  	\draw[-] (7.5,0) -- (8,-2.5);

  	\node [left] at (0,0) { $P$};
  	\node [left] at (3.5,0.2)  { $\langle e_1, e_3\rangle$};
  	\node [left] at (7.5,0.2)  { $\langle e_1,e_2+ e_3\rangle$};
  	\node [right] at (4.7,1.4)  { $\langle e_3, e_4\rangle$};
  	\node [right] at (4.1,2.2)  { $\langle e_3, e_2+e_4\rangle$};
  	\node [right] at (3.1,2.8)  { $\langle e_1+e_3, e_4\rangle$};
  	\node [left] at (3.2,2.8)  { $\langle e_1+e_3, e_2+e_4\rangle$};
  	\node [left] at (1.4,-1)  { $\langle e_3, e_1+e_4\rangle$};
  	\node [left] at  (2,-1.8)  { $\langle e_3, e_1+e_2+e_4\rangle$};
  	\node [left] at (2.8,-2.4)  { $\langle e_1+e_3, e_1+e_4\rangle$};
  	\node [left] at (5,-2.9)  { $\langle e_1+e_3, e_1+e_2+e_4\rangle$};
  	
  	\node [right] at (9,1.2)  { $\langle e_2+e_3, e_4\rangle$};
  	\node [right] at (8.2,2.1)  { $\langle e_2+e_3, e_2+e_4\rangle$};
  	\node [right] at (7.9,2.8)  { $\langle e_1+e_2+e_3, e_4\rangle$};
  	\node [left] at (8,3.2)  { $\langle e_1+e_2+e_3, e_2+e_4\rangle$};
  	\node [left] at (6.8,-.5)  { $\langle e_2+e_3, e_1+e_4\rangle$};
  	\node [left] at  (7.4,-1.9)  { $\langle e_2+e_3, e_1+e_2+e_4\rangle$};
  	\node [left] at (7.9,-2.5)  { $\langle e_1+e_2+e_3, e_1+e_4\rangle$};
  	\node [left] at (9.9,-2.9)  { $\langle e_1+e_2+e_3,e_3+e_4\rangle$};
  	
  	\draw [fill] (0,0) circle [radius=0.08];
  	\draw [fill] (3.5,0) circle [radius=0.08];
  	\draw [fill] (7.5,0) circle [radius=0.08];
  	\draw [fill] (5.5,1.1) circle [radius=0.08];
  	\draw [fill] (1.4,-1) circle [radius=0.08];
  	\draw [fill]  (5,2) circle [radius=0.08];
  	\draw [fill](2,-1.8) circle [radius=0.08];
  	\draw [fill] (4.2,2.5) circle [radius=0.08];
  	\draw [fill] (2.8,-2.4) circle [radius=0.08];
  	\draw [fill] (3,2.5) circle [radius=0.08];
  	\draw [fill] (4,-2.4) circle [radius=0.08];
  	
  	\draw [fill] (9.8,.9) circle [radius=0.08];
  	\draw [fill] (5.4,-.8) circle [radius=0.08];
  	\draw [fill] (9.2,1.8) circle [radius=0.08];
  	\draw [fill](6,-1.7) circle [radius=0.08];
  	\draw [fill] (8.3,2.5) circle [radius=0.08];
  	\draw [fill] (6.8,-2.3) circle [radius=0.08];
  	\draw [fill] (7,2.6) circle [radius=0.08];
  	\draw [fill] (8,-2.5) circle [radius=0.08];
  	\end{tikzpicture}
  	\]

  	Note that any parity check gives rise to a path from $P$ to a point in either $\P$ or $\overline{P}^{(2)}$. For example, the parity check corresponding to the line $L(\langle e_1+e_2\rangle,\langle e_1,e_2,e_4\rangle)$ gives rise to two paths: $(P,\langle e_1+e_2,e_4\rangle)$ and $(P,\langle e_1+e_2,e_2+e_4\rangle)$. The parity check $\omega + \omega_1 + \omega_1^{\prime}$ described above, gives rise to four paths $(P,\langle e_1,e_3\rangle,\langle e_3,e_4\rangle)$, $(P,\langle e_1,e_3\rangle,\langle e_3,e_1+e_4\rangle)$, $(P,\langle e_1,e_2+e_3\rangle,\langle e_2+e_3,e_4\rangle)$, and $(P,\langle e_1,e_2+e_3\rangle,\langle e_2+e_3,e_1+e_4\rangle)$. This is the reason we studied paths in the previous section. If we fix the flag $0 \subset \langle e_1\rangle \subset \langle e_1,e_2\rangle\subset \langle e_1,e_2,e_3\rangle \subset V$, then both paths $(P,\langle e_1+e_2,e_4\rangle)$ and $(P,\langle e_1+e_2,e_2+e_4\rangle)$ have the same $1$-tuples, namely $\mathbf{r}=(2)$ and $\mathbf{s}=(1).$ The four paths coming from the parity check $\omega + \omega_1 + \omega_1^{\prime}$ have the same $2$-tuples, namely $\mathbf{r}=(2,1)$ and $\mathbf{s}=(1,2)$. Note that both $\mathbf{r}$ and $\mathbf{s}$ are strictly monotonous.
  	It is possible to consider other parity checks of weight five, for example one obtained by combining the lines $L(\langle e_1\rangle, \langle e_1,e_2,e_4\rangle )$, $L(\langle e_4\rangle, \langle e_1,e_2,e_4\rangle )$, and $L(\langle e_2+e_4\rangle, \langle e_1,e_3,e_2+e_4\rangle ).$ Also this parity check would give rise to four paths, one of them being $(P,\langle e_1,e_4\rangle,\langle e_3,e_4\rangle).$ The $2$-tuples for these four paths are also the same, namely $\mathbf{r}=(2,1)$ and $\mathbf{s}=(2,2)$. Note that the strict monotonicity is not satisfied in $\mathbf{s}$. We see that in this example, we can get a maximal set of parity checks orthogonal on $P$ by studying paths starting at $P$ of varying lengths with strict monotonous $\mathbf{r}$ and $\mathbf{s}$ tuples.
  	This is the reason we studied paths where both $\mathbf{r}$ and $\mathbf{s}$ are strictly monotonous in Theorems \ref{thm: strictineq} and \ref{thm: uniquepath}.
  	\end{example}

  In the next theorem we show that the observations from the previous example can be generalized for any code $\C$. Recall that for $Q \in \Picl\setminus\Piicl$, we defined the $i$-tuples $\mathbf{r}(Q)$ and $\mathbf{s}(Q)$ in Remark \ref{rem:rQ}. In view of Theorem \ref{thm: uniquepath} these are the $i$-tuples of the unique path from $P$ to $Q$ having strictly monotonous $i$-tuples. Also recall that we throughout are working with a fixed complete flag of $V$, namely \[
  (0)=\mathcal{U}_0\subset \mathcal{U}_1 \subset \mathcal{U}_2\subset \cdots\mathcal{U}_{\ell-1}\subset\mathcal{U}_{\ell}=P=\mathcal{W}_{\ell}\subset \mathcal{W}_{\ell+1}\subset \cdots\mathcal{W}_{m-1}\subset\mathcal{W}_{m}=V.
  \]

  \begin{theorem}
  	\label{thm: paritychecks}
  	Let $\ell,m$ be positive integers satisfying $\ell \le m$ and $\C$ be the corresponding Grassmann code. Then for every $1\le i\le \ell$ there exists a set $\mathcal J_i$ of  $J_i:=\left\lfloor \dfrac q 2\right\rfloor ^i q^{i^2-i} {\ell\brack i}_q{m-\ell\brack i}_q$ many parity checks of $\C$ of Hamming weight $1+ 2^i$ such that:
  	\begin{enumerate}
  		\item For any $\omega \in \mathcal J_i$, the support of $\omega$ consists of $P$ and $2^i$ points from the set $\Picl\setminus\Piicl$.
  		\item For any $\omega \in \mathcal J_i$ and $Q, Q^\prime\in\supp(\omega) \backslash\{P\}$, we have
  		\[
  		\mathbf{r}(Q)=\mathbf{r}(Q^\prime) \ \text{ and } \ \mathbf{s}(Q)=\mathbf{s}(Q^\prime).
  		\]
  		\item For any two distinct $\omega,\omega^\prime \in \mathcal J_i$ we have $\supp(\omega) \cap \supp(\omega^\prime)=\{P\}$.
  		\item For any $i$-tuples $(r_1,\ldots, r_i)$ and $(s_1,\ldots, s_i)$ satisfying $\ell\ge r_1>\cdots>r_i\geq 1$ and $1\le s_1<\cdots< s_i\le m-\ell$, there exist exactly $\left\lfloor \dfrac q 2\right\rfloor ^i\prod_{j=1}^{i}q^{\ell-r_j+s_j-1}$ parity checks $\omega$ in $\mathcal J_i$, such that:
  		
  		\noindent
  		$\text{for any } Q \in \supp(\omega)\backslash\{P\}, \, \mathbf{r}(Q)=(r_1,\ldots, r_i) \text{ and } \mathbf{s}(Q)=(s_1,\ldots, s_i).$
  		\end{enumerate}
  \end{theorem}

  \begin{proof}
  	The proof is by induction on $i$. Assume $i=1$. For each line, we obtain $\lfloor q/2 \rfloor$ parity checks of weight three as follows. We partition the points on the line distinct from $P$ into $\lfloor q/2 \rfloor$ subsets of cardinality two and, if $q$ is odd, a subset containing only one point. For each such subset, say $\{Q,R\}$ there is a parity check $\omega$ such that $\supp(\omega)=\{P,Q,R\}$, by  Theorem \ref{thm: BP}.
  	Since there are ${\ell\brack 1}_q{m-\ell\brack 1}_q$ lines in $\Glm$ through $P$, we obtain a set $\mathcal J_1$ with $\lfloor q/2 \rfloor {\ell\brack 1}_q{m-\ell\brack 1}_q$ parity checks. It is clear that these parity checks satisfy items (1) and (3).
  	
  	We now show that for any two given points $Q,Q^\prime$, not equal to $P$, on a line $L(U,W)$ through $P$ it holds that $r_1(Q)=r_1(Q^\prime)$ and $s_1(Q)=s_1(Q^\prime)$. From this item (2) will follow.
  	Since $U=P\cap Q= P\cap Q^\prime$ and $\mathcal U_{t} \subseteq P$ for every $0 \le t \le \ell$, we get
  	\begin{eqnarray*}
  	r_1(Q)&=&\max\{j:\mathcal{U}_{j-1}\subseteq Q\}\\
  	&=&\max\{j:\mathcal{U}_{j-1}\subseteq P\cap Q\}\\
  	&=&\max\{j:\mathcal{U}_{j-1}\subseteq P\cap Q^\prime\}\\
  	&=&\max\{j:\mathcal{U}_{j-1}\subseteq Q^\prime\}\\
  	&=& r_1(Q^\prime).
  	\end{eqnarray*}
  	Similarly, as $W=P + Q= P+ Q^\prime$ and $P \subseteq \mathcal W_{\ell+t}$ for every $0 \le t \le m-\ell$ , we get
  \begin{eqnarray*}
  s_1(Q)&=&\min\{j:Q\subset \mathcal{W}_{\ell+j}\}\\
  &=&\min\{j:P + Q\subset \mathcal{W}_{\ell+j}\}\\
  &=&\min\{j:P + Q^\prime\subset \mathcal{W}_{\ell+j}\}\\
   &=&\min\{j:Q^\prime\subset \mathcal{W}_{\ell+j}\}\\
  &=& s_1(Q^\prime).
  	\end{eqnarray*}
  To complete the induction basis, we show item (4). Let $\ell\ge r_1\geq 1$ and $1\le s_1\le m-\ell$ be given. Consider all $(\ell-1)$-dimensional $U\subset P$ such that $\mathcal{U}_{r_1-1}\subset U$ but $\mathcal{U}_{r_1}\nsubseteq U$. There are exactly ${\ell-r_1+1\brack 1}_q-{\ell-r_1\brack 1}_q= q^{\ell-r_1}$ such spaces. Similarly, consider all $(\ell+1)$-dimensional spaces $W$ satisfying $P\subset W\subset \mathcal W_{\ell+s_1}$ but $W\nsubseteq \mathcal W_{\ell+s_1-1}$. There are exactly ${s_1\brack 1}_q-{s_1-1\brack 1}_q= q^{s_1-1}$ such $W$. Now take any point $Q$ distinct from $P$ on a line $L(U,W)$, with $U$ and $W$ chosen as above. Then by construction $r_1(Q)=r_1$, since $\mathcal U_{r_1-1} \subset U \subset Q$, while $\mathcal U_{r_1} \subset Q$ would imply that $\mathcal U_{r_1} \subset Q \cap P=U$ using that $\mathcal U_{r_1} \subset P$. Similarly $s_1(Q)=s_1.$ Is either $U$ contains $\mathcal U_{r_1}$ or $W$ is contained in $\mathcal W_{\ell-s_1-1}$, then for any point $Q$ on $L(U,W)$, we have $r_1(Q) > r_1$ or $s_1(Q) < s_1.$ Hence no other parity checks in $\mathcal J_1$ satisfy the requirements from item (4). This completes the proof of item (4).

  Now we consider the induction step. Assume that $i\ge 2$ and that the theorem is true for $i-1$. Let $\mathbf{r}=(r_1,\dots,r_i)$ and $\mathbf{s}=(s_1,\dots,s_i)$ be two given $i$-tuples satisfying $\ell\ge r_1>\cdots>r_{i} \ge 1$ and $1\le s_1<\cdots< s_{i} \le m-\ell$. Then $\ell\ge r_1>\cdots>r_{i-1}> 1$ and $1\le s_1<\cdots< s_{i-1}< m-\ell$. By the induction hypothesis, we know that there exist precisely $\lfloor q/2 \rfloor^{i-1}\prod_{j=1}^{i-1}q^{\ell-r_j+s_j-1}$ parity checks $\omega$ in $\mathcal J_{i-1}$ with $(i-1)$-tuples $(r_1,\dots,r_{i-1})$ and $(s_1,\dots,s_{i-1})$. For any of these parity checks, we are going to construct a set $\mathcal J_i(\mathbf{r},\mathbf{s})$ consisting of exactly $\lfloor q/2 \rfloor q^{\ell-r_i+s_i-1}$ parity checks of weight $1+2^i$ satisfying (1), (2), (3), and having $i$-tuples $\mathbf{r}$ and $\mathbf{s}$.

  Choose $Q_{i-1}\in\supp(\omega)\setminus\{P\},$ then by Theorems \ref{thm: strictineq} and \ref{thm: uniquepath} there exists a unique path $\mathcal P_{i-1}=(P,Q_1,\dots,Q_{i-1})$ from $P$ to $Q_{i-1}$ such that $\mathbf{r}(\mathcal P_{i-1})=(r_1,\dots,r_{i-1})$ and $\mathbf{s}(\mathcal P_{i-1})=(s_1,\dots,s_{i-1})$.  We claim that there exist $q^{\ell-r_i+s_i-1}$ many lines $L(U,W)$ in $\Glm$ passing though $Q_{i-1}$ such that for any point $Q_i$ on $L(U,W)$ different from $Q_{i-1}$, the sequence $\mathcal P_i=(P,Q_1,\dots,Q_{i-1},Q_i)$ is a path from $P$ to $Q_i$ satisfying $\mathbf{r}(\mathcal P_{i})=(r_1,\dots,r_{i})$ and $\mathbf{s}(\mathcal P_{i})=(s_1,\dots,s_{i})$.
    First of all, if $L(U,W)$ is a line through $Q_{i-1}$ such that for one point $Q_i$ on $L(U,W)$ different from $Q_{i-1}$, the sequence $\mathcal P_i=(P,Q_1,\dots,Q_{i-1},Q_i)$ is a path from $P$ to $Q_i$ satisfying $\mathbf{r}(\mathcal P_{i})=(r_1,\dots,r_{i})$ and $\mathbf{s}(\mathcal P_{i})=(s_1,\dots,s_{i})$, then the same is true for all the other points on $L(U,W)$ as well. Indeed, if $Q_i^\prime$ is another point on $L(U,W)$, then somewhat similarly as in the induction basis, one obtains
    \begin{eqnarray*}
    	r_i=r_i(Q_i)&=&\max\{j:\mathcal{U}_{j-1}\subseteq Q_i\}\\
    	&=&\max\{j:\mathcal{U}_{j-1}\subseteq Q_{i-1}\cap Q_i\} \quad \text{since } \mathcal U_{r_{i}-1} \subseteq \mathcal U_{r_{i-1}-1} \subseteq Q_{i-1} \\
    	&=&\max\{j:\mathcal{U}_{j-1}\subseteq Q_{i-1}\cap Q_i^\prime\} \quad \text{since } Q_{i-1}\cap Q_i=U=Q_{i-1}\cap Q_i^\prime\\
    	&=&\max\{j:\mathcal{U}_{j-1}\subseteq Q_i^\prime\} \quad \text{since } \mathcal U_{r_{i}(Q_i^\prime)-1} \subseteq \mathcal U_{r_{i-1}-1} \subseteq Q_{i-1}\\
    	&=& r_i(Q_i^\prime).
    \end{eqnarray*}
    Similarly one obtains $s_i(Q_i^\prime)=s_i.$

To obtain the number of possible lines $L(U,W)$ it is now enough to count the number of points $Q_i$ in $\Glm$ satisfying:
  \begin{enumerate}[label=(\alph*)]
  	\item $\dim (Q_{i-1}\cap Q_i)=\ell-1$,
  	\item $\dim (P\cap Q_i)=\ell-i$,
  	\item $\mathbf{r}(Q_i)=(r_1,\ldots, r_i)$, i.e $\mathcal{U}_{r_i-1}\subseteq Q_i$ but $\mathcal{U}_{r_i}\nsubseteq Q_i$, and
  	\item $\mathbf{s}(Q_i)=(s_1,\ldots, s_i)$, i.e. $Q_i\subseteq \mathcal{W}_{\ell+s_i}$ but $Q_i\nsubseteq \mathcal{W}_{\ell+s_i-1}$.
  \end{enumerate}
  Indeed, the first two condition are equivalent to saying that $\mathcal P_i=(P,Q_1,\dots,Q_{i-1},Q_i)$ is a path from $P$ to $Q_i$, while the last two conditions guarantee that $\mathbf{r}(\mathcal P_{i})=(r_1,\dots,r_{i})$ and $\mathbf{s}(\mathcal P_{i})=(s_1,\dots,s_{i})$.
%
%
Since $r_i<r_{i-1}$ and $s_i>s_{i-1}$, we have
 \begin{equation}
 \label{eq: UinQi}
  \mathcal{U}_{r_i-1}\subset \mathcal{U}_{r_i}\subseteq\mathcal{U}_{r_{i-1}-1}\subseteq Q_{i-1}\cap P
 \end{equation}
 and similarly
 \begin{equation}
 \label{eq: WinQi}
 P+Q_{i-1}\subseteq \mathcal{W}_{\ell+s_{i-1}}\subseteq\mathcal{W}_{\ell+s_{i}-1}\subset\mathcal{W}_{\ell+s_{i}}.
 \end{equation}
First, we compute the number of possibilities for codimension one spaces $U$ in $Q_{i-1}$, which will play the role of $Q_i\cap Q_{i-1}$, and then the number of possibilities in which to extend $U$ to an $\ell$-dimensional space satisfying $(a)-(d)$.

Keeping equation (\ref{eq: UinQi}) and condition $(c)$ in mind, we have that any such $U$ should satisfy $\mathcal{U}_{r_i-1}\subseteq U$ but $\mathcal{U}_{r_i}\nsubseteq U$. Hence there are ${\ell-r_i+1\brack 1}_q-{\ell-r_i\brack 1}_q=q^{\ell-r_i}$ many choices for $U$.
Given one of these choices for $U$ we choose $Q_i\in\Glm$ containing $U$ and satisfying $Q_i\subseteq \mathcal{W}_{\ell+s_i}$ but $Q_i\nsubseteq \mathcal{W}_{\ell+s_i-1}$. There are ${s_i+1\brack 1}_q-{s_i\brack 1}_q=q^{s_i}$ many possibilities for $Q_i$. We claim that this $Q_i$ satisfies $(a)-(d)$.

By construction $U\subset Q_i\cap Q_{i-1}$ and $Q_i\nsubseteq \mathcal{W}_{\ell+s_i-1}$. Since equation (\ref{eq: WinQi}) implies $Q_{i-1}\subseteq \mathcal{W}_{\ell+s_i-1}$, we see that $Q_i\neq Q_{i-1}$. Hence $Q_i\cap Q_{i-1}=U$ and $\dim (Q_i\cap Q_{i-1})=\ell-1$. This proves $(a)$.

Note that $U\cap P\subsetneq Q_{i-1} \cap P$, since $\mathcal{U}_{r_i}\nsubseteq U$, but $\mathcal{U}_{r_i}\subset Q_{i-1}\cap P$. Hence $\dim (U\cap P)\le \ell-i$. On the other hand $U$ is a hyperplane in $Q_{i-1}$ and $U\cap P= U\cap (Q_{i-1}\cap P)$. Hence $\dim (U\cap P)\ge \dim (Q_{i-1}\cap P)-1=\ell-i$. We conclude $\dim (U\cap P)=\ell-i$. Clearly, $U\cap P\subseteq Q_i\cap P$, from which we see that $\dim(Q_i \cap P) \ge \ell-i.$ We claim equality holds, which will prove (b). By construction $Q_i\subseteq \mathcal{W}_{\ell+s_i}$ but $Q_i\nsubseteq \mathcal{W}_{\ell+s_i-1}$. Hence $P+Q_i\subseteq \mathcal{W}_{\ell+s_i}$ but $P + Q_i\nsubseteq \mathcal{W}_{\ell+s_i-1}$. Since $U\subset Q_{i-1}$, from equation (\ref{eq: WinQi}) we get $P + U\subseteq \mathcal{W}_{\ell+s_i-1}$ and hence we have $P +U \subsetneqq P + Q_i$. Consequently, $\dim (P+U)<\dim (P+Q_i)$. We have seen that $\dim (P\cap U)=\ell-i$ and therefore $\dim (P+U)=\ell+i-1$.
On the other hand, $\dim (P +Q_i)= 2\ell-\dim (P\cap Q_i)$.
This implies $\dim (P\cap Q_i)<\ell-i+1$ and we conclude that $\dim (P\cap Q_i)=\ell-i$. This proves $(b)$.

To prove $(c)$ we need to show that $\mathcal{U}_{r_i-1}\subseteq Q_i$ but $\mathcal{U}_{r_i}\nsubseteq Q_i$. The first part is clear as $\mathcal{U}_{r_i-1}\subset U\subseteq Q_i$. For the second part note that if $\mathcal{U}_{r_i}\subseteq Q_i$, then from equation (\ref{eq: UinQi}) we get $\mathcal{U}_{r_i}\subseteq Q_i\cap Q_{i-1}=U$. However by construction $\mathcal{U}_{r_i}\nsubseteq U$. Hence $\mathcal{U}_{r_i}\nsubseteq Q_i$. This completes the proof of $(c)$.

Finally, $(d)$ follows by construction of $Q_i$ as  $Q_i\subseteq \mathcal{W}_{\ell+s_i}$ but $Q_i\nsubseteq \mathcal{W}_{\ell+s_i-1}$.

Combining the above, we see that there exist $q^{\ell-r_i+s_i}$ possibilities for $Q_i$. Hence there exist a set $\mathcal L(Q_{i-1},r_i,s_i)$ of $q^{\ell-r_i+s_i-1}$ lines through $Q_{i-1}$ with the desired properties. We fix an enumeration of these $q^{\ell-r_i+s_i-1}$ lines.
If we choose another point $Q_{i-1}^\prime \in\supp(\omega)\setminus\{P\}$, we can use the argument to get a set $\mathcal L(Q_{i-1}^\prime,r_i,s_i)$ of $q^{\ell-r_i+s_i-1}$ lines $L(U^\prime,W^\prime)$ in $\Glm$ through $Q_{i-1}^\prime$ such that for any point $Q_i^\prime$ on $L(U^\prime,W^\prime)$ different from $Q_{i-1}^\prime$, the corresponding sequence $\mathcal P^\prime_i=(P,Q_1^\prime,\dots,Q_{i-1}^\prime,Q_i^\prime)$ is a path from $P$ to $Q_i^\prime$ satisfying $\mathbf{r}(\mathcal P_{i}^\prime)=(r_1,\dots,r_{i})$ and $\mathbf{s}(\mathcal P_{i}^\prime)=(s_1,\dots,s_{i})$. For each point $Q_i^\prime$ we also fix an enumeration of the $q^{\ell-r_i+s_i-1}$ lines.

Now we construct parity checks from $\omega$ as follows: for each $Q_{i-1} \in \supp(\omega)\setminus\{P\}$ and $1 \le a \le q^{\ell-r_i+s_i-1}$, choose, using Theorem \ref{thm: BP}, a parity check $\omega_{a,Q_{i-1}}$ of $\C$ of weight three with support contained in the $a^{\it th}$ line of $\mathcal L(Q_{i-1},r_i,s_i)$, such that the support of $\omega+\omega_{a,Q_{i-1}}$ does not contain $Q_{i-1}$. Like in the induction basis, we will do this in $\lfloor q/2 \rfloor$ different ways using a partition of the points on the $a^{\it th}$ line of $\mathcal L(Q_{i-1},r_i,s_i)$ distinct from $Q_{i-1}$.

Then for each $1 \le a \le q^{\ell-r_i+s_i-1}$, we obtain $\lfloor q/2 \rfloor$ parity checks of the form
$$\eta(a,\omega):=\omega+\sum_{Q_{i-1} \in \supp(\omega)\setminus\{P\}} \omega_{a,Q_{i-1}}.$$
First of all, note that $P \in \supp(\eta(a,\omega))$ and $\supp(\eta(a,\omega))\setminus\{P\} \subset \Picl \backslash \Piicl.$
Also note that by construction, property (2) is satisfied.
Further, $|\supp(\eta(a,\omega))|=1+2^i$. Indeed no lines of $\mathcal L(Q_{i-1},r_i,s_i)$ and $\mathcal L(Q_{i-1}^\prime,r_i,s_i)$ can intersect each other. If they would intersect in a point, say $Q$, there would exist two distinct paths $\mathcal P_i$ and $\mathcal P_i^\prime$ from $P$ to $Q$ both having $i$-tuples $\mathbf{r}$ and $\mathbf{s}$. But this is not possible by Theorem \ref{thm: uniquepath}. Using a similar argument, we obtain that $\supp(\eta(a,\omega) \cap \supp(\eta(a^\prime,\omega^\prime)) = \{P\}$ is $a \neq a^\prime$ or $\omega \neq \omega^\prime$. In particular $\eta(a,\omega)$ and $\eta(a^\prime,\omega^\prime)$ are mutually orthogonal on $P$ if $a \neq a^\prime$ or $\omega \neq \omega^\prime$. If $a=a^\prime$ and $\omega=\omega^\prime$, but we used different sets from the partitions of the same lines in the sets $L(Q_{i-1}^\prime,r_i,s_i)$, then by construction, the points were partitioned after all, the supports of the parity checks intersect in $P$ only.

  This proves (3).
Finally, by construction and using the induction hypothesis, we have for given strictly monotonous $\mathbf{r}=(r_1,\dots,r_i)$ and $\mathbf{s}=(s_1,\dots,s_i)$, found
exactly $\lfloor q/2 \rfloor ^i\prod_{j=1}^{i}q^{\ell-r_j+s_j-1}$ parity checks. Adding over all possible such $i$-tuples and using Lemma \ref{lem:formula}, the result follows.
\end{proof}

\begin{corollary}
 \label{cor: errorcorr}
 Let $\C$ be a Grassmann code and let $P\in \Glm$ be an arbitrary point. There exists a set $\mathcal J$ consisting of $J:=\sum_{i=1}^{\ell}\left\lfloor \dfrac q 2\right\rfloor ^i q^{i^2-i}{\ell\brack i}_q{m-\ell\brack i}_q$ many parity checks for $\C$, which is orthogonal on the coordinate $P$. In particular, using majority logic decoding, we can correct up to $\lfloor\frac{J}{2}\rfloor$ errors.
\end{corollary}

 \begin{proof}
  Let $P\in\Glm$ be an arbitrary point. We define $\mathcal J :=\cup_{i=1}^\ell \mathcal J_i$, where $\mathcal J_i$ are as in Theorem \ref{thm: paritychecks}.
  Choose $1\le i\le \ell$. By Theorem \ref{thm: paritychecks} the set of parity checks $\mathcal J_i$ is orthogonal on $P$.
 Since the support of the parity checks in $\mathcal J_i$ consists of $P$ and a further $2^i$ points in $\Picl\setminus\Piicl$, they are orthogonal to the parity checks from $\mathcal J_t$ for any $t \neq i$.
 This proves that $\mathcal J$ is orthogonal on $P$.
 Using Theorem \ref{thm: paritychecks} again, we see that $|\mathcal J|=\sum_{i=1}^\ell |\mathcal J_i|=J$.
 Now the last part of the theorem follows from Theorem \ref{thm: Maj decod}.
  \end{proof}	

 \begin{remark}
 In the construction of the set $\mathcal J$, many coordinate positions have been used.
 More precisely, since the parity checks in $\mathcal J_i$ have support in $P$ and $2^i$ points of $\Picl \backslash \Piicl$, the total number of points that occur in one of the parity checks in $\mathcal J$ equals:
 $$1+\sum_{i=1}^{\ell}2^i\left\lfloor \dfrac q 2\right\rfloor ^i q^{i^2-i}{\ell\brack i}_q{m-\ell\brack i}_q.$$
If $q$ is even, and in particular for binary Grassmann codes, then equations \eqref{eq:pic0} and \eqref{eq:picl} imply that any point of $\Glm$ occurs in the support of a parity check in $\mathcal J$. Hence the set $\mathcal J$ cannot be extended further for even $q$.
 \end{remark}

\begin{remark}
As Example \ref{ex:g24} shows, the majority logic decoder from Corollary \ref{cor: errorcorr} does not in general decode up to half the minimum distance of $\C$. Let us investigate more closely what happens.
If $\ell=1$, then $C(1,m)$ is an $[n,k,d]=[(q^{m}-1)/(q-1),m,q^{m-1}]$ code. In fact it is a first order projective Reed--Muller code. We have $J=\lfloor q/2 \rfloor { 1 \brack 1}_q{ m-1 \brack 1}_q=\lfloor q/2 \rfloor (q^{m-1}-1)/(q-1).$ Hence in the binary case, we decode up to half the minimum distance, while for large $q$ we can correct up to roughly $d/4$ errors.

More generally, if $\ell$ and $m$ are fixed and $q$ tends to infinity, then it easy to see that $J/d \to 1/2^\ell.$
Hence for large $q$ we can correct up to $d/2^{\ell+1}$ many errors using Corollary \ref{cor: errorcorr}.
If $\ell$ and $q$ are fixed, but $m$ tends to infinity, a direct calculation shows that $\lim_{m \to \infty} J/d=M_q(\ell)/2^\ell$, where $M_q(\ell)$ is as in equation \eqref{eq:Mqell}.
Note that $M_q(\ell)>1$ if $q$ is even, while $M_q(\ell)<1$ if $q$ is odd.
It is not surprising that the case $q$ is even performs better than the odd case, since for even $q$, we have used all points of $\Glm$ is the support of some parity check in $\mathcal J$, while for odd $q$ there are points that do not appear in the support of any parity check in $\mathcal J.$
The following small table gives an impression on what happens for small values of $q$, $\ell$, and $m$.

$
\begin{array}{c|ccccccccccc}
q & 2 &2 &2 &2 &2 &2 &3& 3&3 & 4 & 4\\
\ell & 2 &2 &2 &2 & 3 & 3 & 2 &2 &2 &2 &2\\
m & 4 & 5 & 6 & 7 & 6 & 7 & 4 & 5 &6 & 4 & 5\\
J & 13 & 49 & 185 & 713 & 309 & 2045 & 25 & 169 & 1330 & 114 & 1554\\
d & 16 & 64 & 256 & 1024& 512 & 4096 & 81 & 729 & 6561& 256 & 4096
\end{array}
$
\end{remark}

Note that any one-step majority logic decoder is fast to execute. In our case, the computation of a parity check from $\mathcal J_i$ costs $2^i$ multiplications in $\Fq$. Therefore, to carry out the majority voting for a single coordinate $P \in \Glm$ costs $ \sum_{i=1}^\ell 2^i \left\lfloor \dfrac q 2\right\rfloor ^i q^{i^2-i}{\ell\brack i}_q{m-\ell\brack i}_q\le |\Glm|-1$ multiplications in $\Fq$. Recall that $|\Glm|$ is the length $n$ of the code $\C$. Performing the majority logic decoding on all coordinates therefore takes at most $n(n-1)$ multiplications in $\Fq$.

Kroll--Vincenti have studied permutation decoding for the codes $C(1,4)$, $C(1,5)$, and $C(2,4)$ \cite{KV}.
Ghorpade--Pi\~nero have extended this approach to affine Grassmann codes \cite{BGH}, which are codes that can be seen as Grassmann codes that have been punctured in ${m \brack \ell}_q-q^{\ell(m-\ell)}$ coordinate positions.
The algorithm in \cite{GP} can decode up to $d/\binom{m}{\ell}-1$ errors and although a complexity analysis was not given, it seems that their algorithm uses around $kn^2$ multiplications in $\Fq$.

Let us compare our decoding algorithm with theirs.
First of all, the complexity of our algorithm is slightly better.
Moreover, if $\ell$ and $q$ are fixed, but $m$ tends to infinity, their error-correcting radius will tend to zero, while we have seen that ours tends to $M_q(\ell)/2^{\ell+1}>0$.
Note $\binom{m}{\ell} > 2^{\ell+1}$ for every $\ell \ge 3$, or $\ell=2$ and $m \ge 5$, or $\ell=1$ and $m \ge 5$.
Hence if $\ell$ and $m$ are fixed, but $q$ tends to infinity, our algorithm performs better as well.

 \section{Acknowledgements}
 Peter Beelen would like to acknowledge the support from The Danish Council for Independent Research (DFF-FNU) for the project {\it Correcting on a Curve}, Grant No.~8021-00030B.

 Prasant Singh would like to thank HC Ørsted-COFUND postdoctoral grant {\it Understanding Schubert Codes}. Most of the work of this article was done when he was working under this project at DTU, Denmark. He would also like to express his gratitude to the Indo-Norwegian project supported by RCN, Norway (Project number 280731), and the DST of Govt. of India.


\clearpage


\begin{thebibliography}{AAAA}
	
\bibitem{BGH}
P. Beelen, S. R. Ghorpade and T. Høholdt, Affine Grassmann codes, \emph{IEEE Trans. Inform. Theory}  {\bfseries 56} (2010), 3166--3176.
	
	\bibitem{BGH2}
P. Beelen, S. R. Ghorpade and T. H{\o}holdt, {Duals of affine Grassmann codes and their relatives}, \emph{IEEE Trans. Inform. Theory} \textbf{58} (2012),
3843--3855.

	\bibitem{BP} P. Beelen and F. Pi\~nero, The structure of dual Grassmann codes, \emph{Des. Codes Cryptogr.} {\bfseries 79} (2016), 451--470.


	
	\bibitem{BCN} A.E. Brouwer, A.M. Cohen, A. Neumaier, Distance regular graphs, Springer, 1989.
	
	\bibitem{GP}
	S. R. Ghorpade and F. L. Pi\~nero, Information set and iterative encoding for Affine Grassmann codes, { \it Proceedings 2015 Seventh International Workshop on Signal Design and its Applications in Communications} (IWSDA 2015), 175--179.
	
	\bibitem{GL} S. R. Ghorpade and G. Lachaud, Higher weights of Grassmann codes, \emph{Coding Theory, Cryptography and Related Areas} (Guanajuato, 1998),
	J. Buchmann, T. Hoeholdt, H. Stichtenoth and H. Tapia-Recillas Eds.,
	Springer-Verlag, Berlin, 
	(2000), 122--131.
	
	\bibitem{GK} S. R. Ghorpade and K. V. Kaipa, Automorphism groups of Grassmann codes,
	\emph{Finite Fields Appl.} {\bfseries 23} (2013),  80--102.
	
	\bibitem{GPP}
	 S. R. Ghorpade, A. R. Patil and H. K. Pillai, {Decomposable subspaces, linear sections of Grassmann varieties, and higher weights of Grassmann codes},  {\em Finite Fields Appl.} {\bfseries 15} (2009), 54--68.
	
	
	\bibitem{GT} S. R. Ghorpade and M. A. Tsfasman, {\rm Schubert varieties, linear codes and enumerative combinatorics},
	\emph{Finite Fields Appl.} {\bfseries 11}  (2005), 684--699.
	
	
	
\bibitem{KL}
 S. L. Kleiman and D. Laksov, {\rm Schubert calculus}, \emph{ Amer. Math. Monthly} {\bfseries 79} (1972), 1061--1082.
	
\bibitem{KP} K. Kaipa and H. Pillai, Weight spectrum of codes associated with the Grassmannian $G(3,7)$, \emph{IEEE Trans. Inform. Theory}  {\bfseries 59} (2013), 983--993.
\bibitem{KV}
 H. J. Kroll and R. Vincenti, PD-sets for binary RM-codes and the codes related to the Klein quadric and to the Schubert variety of PG(5,2), \emph{Discrete Math.} {\bfseries 308} (2008), 408–-414.

\bibitem{LC}
Shu Lin and Daniel J. Costello Jr., Error Control Coding, Pearson Prentice Hall, 1983.



	\bibitem{MS}
	F. J. MacWilliams and N. J. A. Sloane, The Theory of Error Correcting Codes, Elsevier, New York, 1977.
	
	
	\bibitem{M} J.L.~Massey, Threshold decoding, Massachusetts Institute of Technology, Research Laboratory of Electronics, Tech. Rep. 410, Cambridge, Mass., 1963.
	
	\bibitem{Manivel}
	L. Manivel, {\rm Symmetric functions, Schubert polynomials and degeneracy loci.} Translated from the 1998 French original by John R. Swallow. SMF/AMS Texts and Monographs, 6. Cours Spécialisés [Specialized Courses], 3. American Mathematical Society, Providence, RI; Société Mathématique de France, Paris, 2001.
	
	\bibitem{MP} M. Pankov, Grassmannians of Classical Buildings, World Scientific, 2010.
	
	\bibitem{Nogin}	D.Yu.~Nogin, {\rm Codes associated to Grassmannians}, \emph{Arithmetic, Geometry and Coding Theory} (Luminy, 1993), R. Pellikaan, M. Perret, S. G. Vl\u{a}du\c{t}, Eds., Walter de Gruyter, Berlin,	(1996), 145--154.	
	
	\bibitem{Nogin1} D. Yu. Nogin, The spectrum of codes associated with the Grassmannian variety $G(3,6)$, \emph{Problems of Information Transmission } {\bfseries 33} (1997), 114--123
	
	
	
	\bibitem{Ryan} C.T.~Ryan, { An application of Grassmannian varieties to coding theory}, \emph{Congr. Numer.} {\bfseries 157} (1987), 257--271.
	
	\bibitem{Ryan2} C.T.~Ryan, { Projective codes based on Grassmann varieties}, \emph{Congr. Numer.} {\bfseries 157} (1987), 273--279.
	\bibitem{SK}
    D. Silva and  F. R.  Kschischang, { On Metrics for Error Correction in Network Coding}, \emph{IEEE Trans. Inform. Theory}  {\bfseries 55} (2009), 5479--5490.
   \bibitem{TVN}
   M. Tsfasman, S. Vlǎduţ and D. Nogin, Algebraic Geometric Codes: Basic Notions, Mathematical Surveys and Monographs, 139. American Mathematical Society, Providence, RI, 2007.
	
\end{thebibliography}
\end{document}